%% file: samplepaper.tex
\documentclass[runningheads]{llncs}
\usepackage{graphicx}
\usepackage{xspace}
\usepackage{ulem}
\usepackage{latexsym}
\usepackage{amsmath}
\usepackage{adjustbox}
\usepackage{subcaption}
\usepackage{turnstile}

\newcommand{\ecalculus}{$\forall$\textsf{Exp+Res}\xspace}
\newcommand{\irc}{\textsf{IR-calc}\xspace}
\newcommand{\irmc}{\textsf{IRM-calc}\xspace}

\newcommand{\qrc}{\textsf{Q-Res}\xspace}

\newcommand{\qurc}{\textsf{QU-Res}\xspace}

\newcommand{\qrat}{\textsf{QRAT}\xspace}

\newcommand{\lqrc}{\textsf{LD-Q-Res}\xspace}
\newcommand{\lquprc}{\textsf{LQU$^+$-Res}\xspace}
\newcommand{\lqurc}{\textsf{LQU-Res}\xspace}
\newcommand{\mergeres}{\textsf{MRes}\xspace}
\newcommand{\mergeresR}{\textsf{MRes-}$\mathcal{R}$\xspace}
\newcommand{\mergeresT}{\textsf{MRes-}$\mathcal{T}$\xspace}
\newcommand{\mergeresM}{\textsf{MRes-}$\mathcal{M}$\xspace}
\newcommand{\crn}{\textsf{CR}$_n$\xspace}

\newcommand{\qcp}{\textsf{CP+}$\forall$\textsf{red}\xspace}
\newcommand{\sqcp}{\textsf{SemCP+}$\forall$\textsf{red}\xspace}
\newcommand{\qef}{\textsf{eFrege+}$\forall$\textsf{red}\xspace}
\newcommand{\qeff}{\textsf{eFrege}\xspace}
\newcommand{\qff}{\textsf{Frege}}
\newtheorem{induction}{Induction}
\newcommand{\comment}[1]{}

\usepackage{tikz}
\usetikzlibrary{shapes,snakes,automata, matrix, arrows.meta, positioning}

\usetikzlibrary{fadings}
\tikzstyle{uedge}=[draw=blue!50!red]
\tikzstyle{fedge}=[draw=blue]
\tikzstyle{iedge}=[draw=red]
\tikzstyle{redge}=[draw=green!50!black]
\tikzstyle{rnode}=[draw,inner sep=2pt,color=black]
\tikzstyle{tnode}=[circle,minimum width=3pt,fill,inner sep=0pt]
\tikzstyle{dotnode}=[circle,minimum width=2pt,fill,inner sep=0pt]
\tikzstyle{labn}=[font=\sffamily,circle,fill=white,inner sep=1pt,draw=black]
\tikzstyle{legn}=[font=\scriptsize]
\tikzstyle{reln}=[circle,fill=white,inner sep=.4pt,draw=black]
\tikzstyle{oreln}=[circle,fill=white,inner sep=.4pt,draw=black!50,solid]
\tikzstyle{oree}=[thick,draw=black!50,densely dashed]
\tikzstyle{ree}=[thick,draw=black]
\tikzstyle{calcn}=[rectangle%
                   ,rounded corners=2mm%
                   ,draw=black%
                   ,top color=black!10,bottom color=black!10,%
                   minimum height=.5cm,minimum width=.5cm,inner sep=5pt%
                 ]
\tikzset{cross/.style={cross out, draw=black, fill=none, minimum size=.4cm, inner sep=0pt, outer sep=0pt}, cross/.default={1pt}}
\usepackage[hidelinks]{hyperref}

\begin{document}
\title{Extending Merge Resolution to a Family of Proof Systems
}
%
%
\author{Sravanthi Chede
\and
Anil Shukla
}
\authorrunning{S.~Chede and A.~Shukla}
%
\institute{Department of Computer Science and Engineering, IIT Ropar, India\\
\email{\{sravanthi.20csz0001,anilshukla\}@iitrpr.ac.in}
}
\maketitle              
\begin{abstract}
Merge Resolution (\mergeres~\cite{mres_paper}) is 
a recently introduced proof system for false QBFs. Unlike other known QBF proof systems, it builds winning strategies for the universal player within the proofs. Every line of this proof system consists of existential clauses along with countermodels. \mergeres stores the countermodels as merge maps.
Merge maps are deterministic branching programs in which isomorphism checking is efficient as a result \mergeres is a polynomial time verifiable proof system.\vspace{0.1cm}

In this paper, we introduce a family of proof systems \mergeresR in which, the information of countermodels are stored in any pre-fixed complete representation $\mathcal{R}$, instead of merge maps. Hence corresponding to each possible complete representation $\mathcal{R}$, we have a sound and refutationally complete QBF-proof system in \mergeresR. To handle arbitrary representations for the strategies, we introduce consistency checking rules in \mergeresR instead of isomorphism checking in \mergeres.
As a result these proof systems are not polynomial time verifiable.
Consequently, the paper shows that using merge maps is too restrictive and can be replaced with arbitrary representations leading to several interesting proof systems.\vspace{0.1cm}

The paper also studies proof theoretic properties of the family of new proof systems \mergeresR. We show that \qef simulates all valid refutations from proof systems in \mergeresR. Since proof systems in \mergeresR may use arbitrary representations, in order to simulate them, we first represent the steps used by the proof systems as a new simple complete structure. As a consequence, the corresponding proof system belonging to \mergeresR is able to simulate all proof systems in \mergeresR.
Finally, we simulate this proof system via \qef using the ideas from~\cite{Chew2021}.\vspace{0.1cm}

On the lower bound side, we lift the lower bound result of regular \mergeres~(\cite{BeyersdorffBMPS20}) for all regular proof systems in \mergeresR. 
To be precise, we show that the completion principle formulas from~\cite{JanotaM15} which are shown to be hard for regular \mergeres in~\cite{BeyersdorffBMPS20}, are also hard for any regular proof system in \mergeresR. 
Thereby, the paper lifts the lower bound of regular \mergeres to an entire class of proof systems, which use some complete representation, including those undiscovered, instead of merge maps. 

\end{abstract}

\section{Introduction}
Proof complexity is a sub-branch of computational complexity in which the main focus is to understand the complexity of proving (refuting) theorems (contradictions) in various proof systems. Informally, a proof system is a polynomial time computable function which maps proofs to theorems. Several propositional proof systems like resolution~\cite{Rob63}, Cutting planes \cite{CookCT87}, and Frege \cite{frege1879} have been developed for proving (refuting) propositional formulas. The relative strength of these proof systems has been well studied~\cite{Segerlind07}. In the literature, several proof systems which are not polynomial time computable (verifiable) have also been well studied. For example, semantic cutting planes~\cite{FilmusHL16}. 

Quantified Boolean formulas (QBFs) extend propositional logic by adding quantifications $\exists$ (there exists) and $\forall$ (for all) to the variables. Several QBF proof systems like \qrc~\cite{KBKF95}, \qurc~\cite{Gelder12}, \lqrc~\cite{Balabanov12}, \ecalculus~\cite{JM15}, \irc, and \irmc~\cite{BCJ14} have been proposed in the literature. These are all either CDCL (Conflict Driven Clause Learning)-based, or expansion-based QBF proof systems. Cutting planes proof systems has also been extended for QBFs (\qcp)~\cite{BeyersdorffCMS18}.

A new proof system Merge resolution (\mergeres)~\cite{mres_paper} has been developed recently. It follows a different QBF-solving approach. In \mergeres, winning strategies for the universal player are explicitly represented within the proof in the form of deterministic branching programs, known as merge maps~\cite{mres_paper}. \mergeres builds partial strategies at each line of the proof such that the strategy at the last line (corresponding to the empty clause) forms the complete countermodel for the input QBF. As a result, \mergeres admits strategy extraction by design. 
Before applying the refutation rules, \mergeres needs to check the partial strategies of the hypothesis to be isomorphic. Note that the isomorphism checking in `merge maps' is efficient, hence \mergeres is a polynomial time verifiable proof system.

In this paper, we extend \mergeres to a family of sound and refutationally complete QBF proof systems \mergeresR. We observe that the representation of strategies in the proofs as merge maps is not relevant for the soundness and completeness of the proof system. Strategies can be depicted by any complete representation and by slightly modifying the refutation rules to include arbitrary complete representations, the soundness and completeness of the proof system remains intact. 
To be precise, we change the isomorphism checking rule in \mergeres to `consistency' checking (Section~\ref{sec:mresrNotations}) defined for Dependency Quantified Boolean Formulas (DQBFs) in \cite{davis_putnam}. This leads to the definition of a new proof system (Say $\mathcal{P}$) for each complete representation. All these new proof systems together form the family of proof systems denoted by \mergeresR. However, due to the consistency checking rules, the proof systems in \mergeresR are not polynomial time verifiable. 
In literature, many interesting non-polynomial time verifiable proof systems have been studied, for example, semantic cutting planes for QBFs (\sqcp)~\cite{BeyersdorffCMS18}.
Because of the introduction of such powerful consistency checking rules, proof systems in \mergeresR allow a few forbidden resolution steps of \mergeres (ref.~Example~\ref{eg:motivation}) 


The paper also studies in detail the strength of these new proof systems. We show that \qef is powerful enough to simulate valid refutations of the proof systems in \mergeresR.
Since these systems admit strategy extraction by design, we show the said simulations by using the ideas from~\cite{Chew2021}. 
Furthermore, the paper lifts the lower bound results from~\cite{BeyersdorffBMPS20} of regular \mergeres to every regular proof system $\mathcal{P} \in$ \mergeresR. We explain our contributions in detail in the following section.

\subsection{Our Contributions}
\begin{enumerate}
    \item \textbf{Introducing a new family of proof systems \mergeresR:} 
    \mergeres~\cite{mres_paper} uses merge maps to store the countermodels within proofs.
    We observe that merge maps are not important for the soundness and completeness of the proof system. They just make the proof system polynomial time verifiable. However, at the same time they are too restrictive. In this paper, we propose a family of proof systems \mergeresR, one for each arbitrary complete representation of strategies into proofs (instead of merge maps). In order to make these proof systems sound and complete, we only need to modify the rules of \mergeres slightly. To be precise, we check the consistency relation instead of isomorphism among the strategies before applying the resolution rules (ref.~Section \ref{section_mresr}). 
    
    Although, this modification makes the proof systems not polynomial time verifiable; however, makes them very interesting, since the representations of strategies can be arbitrary. We only need that the representations be complete, in the sense that every finite function has at least one representation in the same. We need this for proving the completeness of our proposed proof systems (in Claim~\ref{claim_4}).
    
    To be precise, for proving completeness of \mergeresR, we consider the \mergeresM proof system in \mergeresR which uses merge maps as the representation for strategies. We then prove that \mergeresM system p-simulates the \mergeres proof system (which is known to be complete) by showing that every rule of \mergeres is also valid in \mergeresM
    (Theorem~\ref{theorem_3}). We then show
    how any \mergeresM-proof can be (non-efficiently) converted to $\mathcal{P}$-proof for any $\mathcal{P} \in$ \mergeresR (Claim~\ref{claim_4}). 
    
    The soundness proof of \mergeresR follows from proving that every line of the proof gives a partial falsifying strategy for the universal player. 
    
    \item \textbf{Proving a lower bound for Regular \mergeresR:} The Completion Formulas \crn were first introduced in~\cite{JanotaM15}, to show that level-ordered \qrc cannot p-simulate \ecalculus. They were also used to show that level-ordered \qrc cannot p-simulate tree-like \qrc~\cite{MahajanS16}. It has been shown recently in~\cite{BeyersdorffBMPS20}, that \crn formulas are even hard for regular \mergeres. In this paper, we lift this lower bound of the Completion Formulas \crn to all regular proof systems in \mergeresR. That is, we show that any regular proof system $\mathcal{P} \in $ \mergeresR, takes exponential time to refute the \crn formulas. 
    
    For this lower bound proof, we mostly follow the proof from \cite[Theorem 9]{BeyersdorffBMPS20} where they have used the fact that most of the clauses in the \mergeres proof are going to be free of all the literals from right (in quantifier prefix) of the only universal variable $z$ of \crn. They showed that the number of such clauses are exponential in $n$ proving the required lower bound.

    We established the similar argument for every regular proof system in \mergeresR.
    In~\cite[Theorem 9]{BeyersdorffBMPS20}, the major part of the proof relied on the fact that \mergeres uses isomorphism so they could rule out the variables not in one hypothesis merge map of a resolution step, as also not to be present in the other. However, this is not the case in \mergeresR. So we provide a new Claim (ref.~Claim \ref{claim_7}) that even though \mergeresR insists on consistency rather than isomorphism, the above property holds.
    That is, the clauses in \crn make it such that the variables not in one hypothesis strategy cannot be present in the other as well when consistency is maintained in the resolution steps. 

    \item \textbf{\qef simulates \mergeresR:} We show that \qef simulates valid refutations in every proof system belonging to \mergeresR. 
    Since proof systems in \mergeresR can use arbitrary representations, simulating the same is difficult even for powerful proof systems. However \mergeresR family uses simple rules for refuting, which can even be detected by just observing the clauses of the lines, without exploring the representation parts. If one can come up with a complete representation which can represent the rules performed by any \mergeresR proof system, then one can show the required simulation.
    
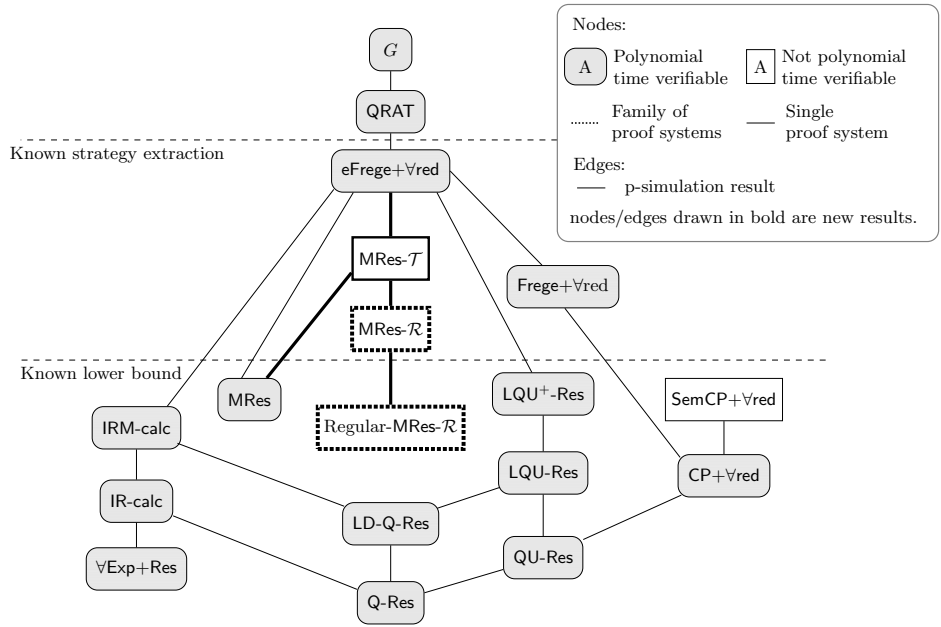
\begin{figure}[h]
\centering
\resizebox{\columnwidth}{!}{%
\input{proof_systems}
}
\caption{\textit{Various QBF proof systems and efficient simulations.} New results and proof systems are written in bold. \mergeresT belongs to \mergeresR. Regular \mergeresR are below the `known lower bound' dashed line, due to Theorem~\ref{thm:reg-meresr-lower}. \mergeresT p-simulates \mergeres due to Proposition~\ref{prop:mrest-simulates-mres}. For the simulations of \mergeres, \irmc, and \lquprc by \qef, and other known simulations refer~\cite[Fig.~1]{Chew2021}.}
\label{fig:systems_diag}
\vspace{-0.4cm}
\end{figure}
     We use this observation and define a new complete representation, denoted as the $T$-representation (ref.~Section \ref{sec:t-representation}). To handle all the \mergeresR rules, we came up with this hybrid representation consisting of both circuits and branching programs. It is capable of representing all the refutation rules allowed in any \mergeresR proof. To achieve this we introduce a new type of node, namely `\#' node, which deals with the new consistency checking property (ref. Fig~\ref{fig:subfig1}). We also show that $T$ representation is a complete representation, therefore the corresponding \mergeresT proof system (Section~\ref{sec:t-representation}) belongs to \mergeresR.

    The idea of the simulation is to convert every valid $\mathcal{P}$-proof ($\mathcal{P} \in$ \mergeresR) efficiently into an \mergeresT-proof as explained in Theorem~\ref{theorem_8}. Then, as \mergeresT admits strategy extraction by definition, we use the idea from \cite{Chew2021} to simulate the valid \mergeresT-proof in \qef. Thereby proving that \qef simulates any valid refutations from proof system in \mergeresR.

    This simulation result is a way forward towards uniform certification in QBF~\cite{Chew2021}. In~\cite{Chew2021}, they used a very distinguished technique that allows \qef to simulate few proof systems that admit strategy extraction. Using this technique, they showed that \qef can simulate proof systems \mergeres, \irmc and \lquprc. 
    We use the same technique and prove that \qef even simulates the family of proof systems \mergeresR. 

    Observe that the proposed \qef simulation algorithm (Section~\ref{sec:qef-simulates-mresr}) of the proof systems $\mathcal{P}$ in \mergeresR assumes that the given $\mathcal{P}$-refutations are valid. As a result one cannot use this simulation algorithm to efficiently verify the correctness of the given $\mathcal{P}$-refutations. That is, even if the resulting \qef proof is valid (which is efficiently verifiable), one cannot infer if or not the initial $\mathcal{P}$-refutation is valid. However, the proposed simulation algorithm always produces a valid \qef proof for a given valid $\mathcal{P}$ proof.
    
    Further, for the current simulation order among QBF proof systems refer Fig~\ref{fig:systems_diag}.

    
\end{enumerate}

\subsection{Organization of the paper}
We present important notations and preliminaries used in this paper in Section~\ref{sec:preliminaries}. In Section~\ref{section_mresr}, we present the new family of proof systems \mergeresR. We prove the soundness and refutational completeness of proof systems in \mergeresR in Section~\ref{sec:soundness-completeness}. In Section~\ref{sec:qef-simulates-mresr}, we show that \qef simulates proof systems in \mergeresR. We establish the lower bound results for every regular proof system in \mergeresR in Section~\ref{sec:lower-bound-regular}. Finally, we conclude and present a few open problems in Section~\ref{sec:conclusion}.

\section{Notations and Preliminaries}\label{sec:preliminaries}
For a Boolean variable $x$, its literals can be $x$ and $\overline{x}$.
A clause $C$ is a disjunction of literals and a conjunctive normal form (CNF) $F$ is a conjunction of clauses.
We denote the empty clause by $\bot$. $vars(C)$ is a set of all variables in $C$ and width$(C) = |vars(C)|$.\\
A \textbf{proof system}~\cite{CookR79} for a non-empty language $L \subseteq \{ 0, 1\}^*$ is a
polynomial time computable function $f : \{ 0, 1\}^* \rightarrow \{ 0, 1\}^*$ such that Range($f$) = $L$. For string $x \in L$, we say a string $w \in \{ 0, 1\}^*$ is an $f$-proof of $x$ if $f(w)$ = $x$. 
A proof system $f$ for $L$ is complete if and only if for every $x \in L$ we have a corresponding $f$-proof for $x$. A proof system $f$ for $L$ is sound if and only if the existence of an $f$-proof for $x$ implies that $x \in L$. Informally, a proof system is a function $f$ which maps proofs to theorems (or contradictions). 

A proof system $f$ p-simulates (polynomially simulates) another proof system $g$, if every $g$-proof of input $x$ can be translated into an $f$-proof for the same input in polynomial time w.r.t the size of the $g$-proof. We denote this as $f \leq_p q$. Proof systems for $L = \text{FQBFs/ TQBFs}$ are said to be QBF proof systems where, FQBFs (TQBFs) denote the set of all false (true) QBFs.\\
{\bf Quantified Boolean formulas:} QBFs are an extension of the propositional Boolean formulas where each variable is quantified with one of $\{ \exists,\forall \}$, the symbols having their general semantic definition of existential and universal quantifier respectively.

In this paper, we assume that QBFs are in closed prenex form with CNF matrix i.e., we consider the form $Q_1 X_1 . . .Q_k X_k$. $\phi(X)$, where $X_i$ are pairwise disjoint sets of variables; $Q_i$ $\in$ \{$ \exists$, $ \forall$\} and $Q_i \neq Q_{i+1}$, and $\phi(X)$ is in CNF form over variables $X = X_1 \cup \dots \cup X_k$, called the matrix of the QBF. We denote QBFs as $\mathcal{F} := Q.\phi$ in this paper, where $Q$ is the quantifier prefix. If $x \in X_i$ then we denote $Q(x)$ to be equal to $Q_i$. For a variable $x$ if $Q(x) = \exists $ (resp.~$Q(x) = \forall$), we call $x$ an existential (resp.~universal) variable.
If a variable $x$ is in the set $X_i$, any $y \in X_j$ where $j<i$, we say that $y$ occurs to the left of $x$ in the quantifier prefix and write $y {\leq}_Q x$. On the other hand, if $j>i$ we say that $y$ occurs to the right of $x$ in the quantifier prefix and write $y \geq_Q x$. The set of existential variables to the left of a universal variable $u$ will be denoted by $L_Q (u)$ in this paper. 

Let $C \in \phi$ and $Q(u) = \forall$, then the `falsifying $u$-literal' is defined to be $0$ if $u \in C$, and $1$ if $\overline{u} \in C$ and $\ast$ if $u \notin vars(C)$. Also, the existential subclause of $C$ is the clause formed by only the existential literals from $C$. 

If $S$ is any set of variables, a complete assignment of $S$ will be an assignment which assigns every variable in $S$ to either $1$ or $0$. Similarly, a partial assignment is an assignment which assigns a subset of variables in $S$ to either $1$ or $0$. Note that the $vars(S)$ that have not been assigned to $1$ or $0$ in a particular partial assignment of $S$ are denoted as having an assignment of `$\ast$'.
We denote $\langle S \rangle$ as the set of all possible complete assignments of $S$ and $\langle \langle S \rangle \rangle$ as the set of all possible partial assignments of $S$. 

{\bf QBFs as a game:}~\cite{AroraBarak09}
Given a QBF $\mathcal{F}= \exists X_1 \forall X_2 \dots \exists X_n. \phi$ we may view it as game between universal and existential player. The rules of the game being that according to the quantification sequence the players assign values to the sets $X_i$s alternatively. At the end, when substituting the complete assignment to all variables in $\phi$, if $\phi$ evaluates to $1$ (resp.~$0$) the existential (resp.~universal) player wins.

For a QBF $Q.\phi$, a \textbf{strategy} of universal player is a decision function that returns the assignment to all universal variables of $Q$, where the decision for each $u$ depends only on the variables in $L_Q(u)$. If $H^u$ is the strategy for the universal variable $u$ then, $vars(H^u)$ is the subset of existential variables from $L_Q(u)$ which are actually used in building the strategy $H^u$.

\textbf{Winning strategy} for the universal player is a strategy which gives an assignment to all universal variables of the given QBF for every possible assignment of existential variables, such that the substitution of this complete assignment falsifies the QBF.
Winning strategy of the universal player is also called a countermodel in case of a false QBF.
A QBF is false iff there exists a winning strategy for the universal player~\cite{AroraBarak09}.

We say that a QBF proof system $f$ admits \textbf{strategy extraction} if for any given valid $f$-proof of a false QBF $\mathcal{F}$, one can compute a winning strategy for the universal player in the time polynomial to the size of the $f$-proof.

As said earlier, strategies are basically decision functions. For the portrayal of the same, many representations can be used like truth tables, directed acyclic graphs (DAGs), merge maps, etc. A complete representation is the one in which every possible finite decision function can be represented.\\
\noindent 
\textbf{Resolution}~\cite{Rob63} is the most studied redundancy rule in both SAT and QBF worlds, we define the same below:
$$\frac{C \vee x \hspace{10mm} D \vee \overline{x}}{C \vee D},$$
where $C,D$ are clauses and $x$ is the pivot variable.
The clause ($C \lor D$) is called the resolvent.
We denote this step as `$res(C\lor x,D\lor \overline{x},x)$' throughout the paper.\\
Next, we define a few QBF proof systems that we require in this paper. 

{\bf {\qrc}}~\cite{KBKF95} is one of a basic QBF proof system. It is an extension of the resolution proof system for QBFs. It allows the resolution rule defined above with the pivot variable being existential. For dealing with the universal variables, it defines a `universal reduction' rule as follows:

The \textbf{Universal Reduction} (UR) rule of \qrc allows dropping of a universal variable $u$ from a clause $C$ in the QBF, provided no existential variable $x \in C$ appears to the right of $u$ in the quantifier prefix.

\subsection{\mergeres}
\mergeres is a proof system for false QBFs introduced in~\cite{mres_paper}.
We describe \mergeres briefly in this section, please refer to \cite{mres_paper} for its formal definition.

For a false QBF $Q.\phi$, an \mergeres refutation will be a sequence of lines of the form $L_i=(C_i,\{M^u_i\})$; where $C_i$ is a clause consisting of only existential literals and $\{ M^u_i \}$ is the set of merge maps of each universal variable $u \in Q$. 

The \textbf{Merge maps} represent the partial strategies for each universal variable at any line. The merge map $M^u_i$ is a decision branching graph with definite strategies $\{0,1,\ast\}$ at the leaves nodes ($\ast$ is used when no strategy for $u$ exists till that line). The intermediate nodes of merge map $M^u_i$ branch on some existential variable (say $x$) $\in L_Q(u)$. That is, if $L_i=res(L_a,L_b,x)$ for some $a,b<i$, then
$M^u_i$ will get connected to $M^u_a$ with an edge label of $\overline{x}$ and to $M^u_b$ with an edge label of $x$.

\noindent 
An important property used in \mergeres refutation rules is defined below:

\textbf{Isomorphism:} Two merge maps are isomorphic if and only if there exists a bijection mapping from the nodes of one to that of another. In other words, two isomorphic merge maps represent exactly the same strategy. 

\noindent Two operations needed for \mergeres refutations are defined below:

\textbf{Select} operation is defined on two merge maps. If they are isomorphic, then it outputs one of them. Or, if one of them is trivial (i.e $\ast$), then it outputs the other.

\textbf{Merge($M^u_a, M^u_b, n, x$)} operation is defined when $a,b<n$, and it returns a new merge map where the new root node is connected to $M^u_a$ with $\overline{x}$ and to $M^u_b$ with $x$. Also, if any common line nodes exist in $M^u_a, M^u_b$, it merges them into a single node.\vspace{0.2cm}\\
Now we define the \mergeres proof system:\\
For a false QBF $Q.\phi$, the \mergeres proof $\Pi := L_1,L_2,...,L_m$ where every line $L_i := (C_i , \{ M^u_i : $for every universal variable $u$ in $Q\} )$ is derived using either an `Axiom' step or a `Resolution' step. In the axiom step, $C_i$ will be the existential subclause of some $C \in \phi$ and every $M^u_i$ will be a leaf node with the falsifying $u$-literal of $C$.
In the resolution step, $C_i$ is obtained from $res(C_a,C_b,x)$ where $x$ is an existential variable and $a,b <i$. For this step to be valid, each $M^u_i$ must either be equal to select($M^u_a,M^u_b$) or if $x <_Q u$ then it can be equal to merge($M^u_a,M^u_b$,i,x).

The final line $L_m$ is the conclusion of $\Pi$, and $\Pi$ is a refutation of $Q.\phi$ iff $C_m = \bot$. 
 $G_\Pi$ be the derivation graph corresponding to $\Pi$ with edges directed from the hypothesis to the resolvent (i.e
from the axioms to the final line).
A refutation $\Pi$ is said to be regular if no leaf-to-root path in $G_\Pi$ has any existential variable $x$ as a pivot more than once.
For some given line $L$, $\Pi_L$ is defined as the sub-derivation of $\Pi$ deriving the line $L$.
\subsection{\qef}
Frege systems are fundamental proof systems of propositional logic. Lines in a Frege proof are formulas inferred from the previous lines via few sound rules. The rules are not important as all Frege systems are p-equivalent, the only condition is that a Frege system needs to be sound and complete.
So w.l.o.g, we can assume that `modus ponens' is the only rule in a Frege system. The modus ponens is defined as: if $A \rightarrow B$ and $A$ are present in the hypothesis then $B$ can be logically implied by the hypothesis.
For a detailed definition and explanation refer~\cite{Jankrajicek19}.

Extended Frege (\qeff)~\cite{efrege_paper} is an extention of Frege systems which allows introduction of new variables not present in previous lines of the proof. This rule allows lines of the form $v \leftrightarrow f(S)$ where $v$ is a new variable and $f$ can be any formula on the set of variables $S$, where $v\notin S$.

For QBFs, \qeff is modified to be \qef (Extended Frege + $\forall$ reduction)~\cite{fregeQBF} which requires that the extention variable must be added in the prefix and quantified to the right of the variables used to define it.
To deal with the universal variables, the universal reduction (UR) rule as defined in \qrc is introduced into \qef. The formal definition is as follows:
$$\frac{L_j}{L_i = L_j[u/B]}~(\forall red)$$
where $L_j$ is some previous line in the \qef proof, $u$ is a universal variable that is also rightmost in the prefix among all variables in $L_j$ and $B$ is the Herband function of $u$~\cite{fregeQBF}.
That is, a universal variable $u$ in a formula can be replaced by $0$ or $1$ when no other variable in that formula are to the right of $u$ in the prefix.\vspace{0.2cm}\\
For the rest of known QBF-proof systems depicted in Fig~\ref{fig:systems_diag}, refer to \cite{Shukla17}.
\noindent We define the following formulas needed for proving lower bound later in the paper.
\subsection{Completion Principle Formulas~\cite{JanotaM15}}\label{CR_formula}
The QBFs \crn are defined as follows:\\
\crn = $\underset{i,j \in [n]}{\exists} x_{ij}$, $\forall z$, $\underset{i \in [n]}{\exists} a_{i}$, $\underset{j \in [n]}{\exists} b_{j}$
$\Big( \underset{i,j \in [n]}{\bigwedge} ( A_{ij} \wedge B_{ij} ) \Big) \wedge L_A \wedge L_B$\\
where,

$A_{ij}= x_{ij} \lor z \lor a_{i}$ \hspace{1.8cm} $B_{ij}= \overline{x_{ij}} \lor \overline{z} \lor b_{j}$

$L_A= \overline{a_1} \lor \dots \lor \overline{a_n}$ \hspace{1.7cm} $L_B= \overline{b_1} \lor \dots \lor \overline{b_n}$ \vspace{0.2cm}

\noindent For any \crn formula, we define the sets
$A := \{ a_1, a_2,...,a_n \}$ and $B:= \{ b_1, b_2,...,b_n \}$ as the set of all $a$, $b$ variables respectively. 

\comment{
\subsubsection{KBKF-$lq[n]$ Formulas\cite{BalabanovWJ14}} These QBFs are a variant of the formulas introduced by Kleine Büning et al.\cite{resolution}\\
KBKF-lq[n]$= \exists d_1,e_1, \forall x_1, \exists d_2,e_2, \forall x_2, . . . , \exists d_n,e_n, \forall x_n, \exists f_1,f_2,...,f_n$

\hspace{2cm}$A_0 \underset{i \in [n]}{\bigwedge} \Big( A_i^d \land A_i^e \Big) \underset{i \in [n]}{\bigwedge} \Big( B_i^0 \land B_i^1 \Big) $\\
where,

$A_0 = \{ \overline{d_1}, \overline{e_1}, \overline{f_1}, . . . , \overline{f_n} \}$

$A_i^d= \{ d_i, x_i, \overline{d_{i+1}}, \overline{e_{i+1}}, \overline{f_1}, . . . , \overline{f_n} \}$ \hspace{0.1cm} $A_i^e= \{ e_i, \overline{x_i}, \overline{d_{i+1}}, \overline{e_{i+1}}, \overline{f_1}, . . . , \overline{f_n} \}$ \hspace{0.02cm} $\forall i \in [n$-$1]$

$A_n^d= \{ d_i, x_i, \overline{f_1}, . . . , \overline{f_n} \}$ \hspace{1.6cm} $A_n^e= \{ e_i, \overline{x_i}, \overline{f_1}, . . . , \overline{f_n} \}$

$B_i^0= \{ x_i, f_i, \overline{f_{i+1}}, . . . , \overline{f_n} \}$ \hspace{1.35cm} $B_i^1= \{\overline{x_i}, f_i, \overline{f_{i+1}}, . . . , \overline{f_n} \}$ \hspace{1.25cm} $\forall i \in [n$-$1]$

$B_n^0= \{ x_i, f_i \}$ \hspace{3.1cm} $B_n^1= \{\overline{x_i}, f_i \}$ \vspace{0.1cm}\\
Shorthand notations:

$A_i = \{ A_i^d, A_i^e \}$ \hspace{3.05cm} $B_i = \{ B_i^d, B_i^e \}$ \hspace{2.75cm} $\forall i \in [n]$}
\section{\mergeresR: A new family of proof systems for false QBFs}\label{section_mresr}
We define a family of proof systems \mergeresR, inspired from the \mergeres proof system. In \mergeres~(\cite{mres_paper}), strategies are built within the proof and are represented by merge maps. We observed that merge maps or any specific representations of strategies are not important for the soundness or completeness of the proof system. Since, isomorphism problem is efficient in merge maps, they make the proof systems polynomial time verifiable. 

Based on this observations, we define a family of \mergeresR where
every proof system $\mathcal{P} \in$ \mergeresR has it's own complete representation to represent the strategies.
To allow the use of arbitrary representations in \mergeresR, we introduce consistency checking rules for strategies which are not as efficient as checking isomorphism for \mergeres. 
As a result, our proof systems are not polynomial time verifiable. However their soundness \& completeness doesn't depend on their representations, which makes them interesting.

We use the idea of consistency checking from~\cite{davis_putnam}, which uses the same for DQBFs. For simplicity, we use the same notations from~\cite{davis_putnam} whenever possible. We begin by defining some important notations and operations needed before actually defining the \mergeresR systems. 
\subsection{Important notations used in \mergeresR}\label{sec:mresrNotations}
To begin, let us define what consistency means for any two assignments of a set of variables.
\begin{definition}[\cite{davis_putnam}]
Let $X$ be any set of variables and $\varepsilon,\delta \in \langle \langle X \rangle \rangle$. We say that $\varepsilon$ and $\delta$ are consistent, denoted by $\varepsilon \simeq \delta$, if for every $x \in X$ for which $\varepsilon(x) \neq \ast$ and $\delta(x) \neq \ast$ we have $\varepsilon(x) = \delta(x)$.
\end{definition}
Let $H_u$ and $H_u'$ be individual strategy functions for the universal variable $u$, the consistency between two strategies is defined as follows:

We say that $H_u$ and $H_u'$ are consistent (written $H_u \simeq H_u'$) when $H_u(\varepsilon) \simeq H_u'(\varepsilon)$ for each $\varepsilon \in \langle \langle L_Q(u) \rangle \rangle$. Recall that $L_Q(u)$ are the existential variables to the left of $u$ in the prefix.
In other words $H_u$ and $H_u'$ are consistent, if the $u$-assignment given by $H_u(\varepsilon)$ and  $H_u'(\varepsilon)$ should be consistent for every possible $L_Q(u)$-assignment $\varepsilon$. 

By a change in notation, we can see (partial) assignments as both functions and sets of literals, i.e. an assignment $\varepsilon$ corresponds to the set of literals it satisfies. 
For example, $\{ x_1, x_2, \overline{x_3}, \overline{x_4} \}$ represents an assignment which sets 1 to the variables $x_1$ and $x_2$ and 0 to $x_3$ and $x_4$.
In this notation as sets of literals, a union ($\cup$) of assignments $\varepsilon,\delta$ is defined when $\varepsilon \simeq \delta$ 
and it is equal to $\varepsilon \cup \delta$.\vspace{0.2cm}\\
We now define a union operation (`$\circ$') on two consistent strategies $H_u$ and $H_u'$.
\begin{definition}[\cite{davis_putnam}]
Given two consistent strategies $H_u$ and $H_u'$ (i.e., $H_u \simeq H_u'$), we define the union strategy $H_u''$ of $H_u$ and $H_u'$, denoted by $H_u'' = H_u \circ H_u'$, as follows: 

\centering$H_u''(\varepsilon)= H_u(\varepsilon) \cup H_u'(\varepsilon)$ for each $\varepsilon \in \langle L_Q(u) \rangle.$
\end{definition}
For example, if $H_u \& H_u'$ be defined as below, then $H_u''= H_u \circ H_u'$ will be:\\  
\hspace*{1cm}\(
H_u = \left\{ 
\begin{array}{l l}
  1 & :\quad x\\
  * & :\quad \overline{x}\\
\end{array} \right.
\) \hspace{1.7cm}
\(
H_u' = \left\{ 
\begin{array}{l l}
  * & :\quad x\\
  0 & :\quad \overline{x}\\
\end{array} \right.
\)\\
\hspace*{3.3cm}\(
H_u'' = \left\{ 
\begin{array}{l l}
  1 \cup * = 1 & :\quad x\\
  * \cup 0 = 0 & :\quad \overline{x}\\
\end{array} \right.
\)
\vspace{0.2cm}\\
We now define a if-else operation (`$\bowtie$') on any two strategies $H_u$ and $H_u'$.
\begin{definition}[\cite{davis_putnam}]\label{if_else_def}
Given any two strategies $H_u$ and $H_u'$ and an existential variable $x$, we define the if-else operation of $H_u$ and $H_u'$ on $x$ to give the strategy $H_u''$, denoted by $H_u'' = H_u \overset{x}{\bowtie} H_u'$, for every $\varepsilon \in \langle L_Q(u) \rangle$ as follows: \vspace{0.2cm}\\
\hspace*{3.3cm}\(
H_u''(\varepsilon) = \left\{ 
\begin{array}{l l}
  H_u(\varepsilon) & :\quad \varepsilon(x)=1\\
  H_u'(\varepsilon) & :\quad \varepsilon(x)=0\\
\end{array} \right.
\)
\end{definition}
\comment{
on previous inputs \( 
H_u'' = \left\{
\begin{array}{l l}
  1 & :\quad x x_1\\
  \ast & :\quad x \overline{x_1}\\
  \ast & :\quad \overline{x} x_1\\
  0 & :\quad \overline{x} \overline{x_1}\\
\end{array} \right.
\)}
For example, if $H_u \& H_u'$ be defined as below, then $H_u'' = H_u \overset{x}{\bowtie} H_u'$ will be:
\hspace*{1cm}\(
H_u = \left\{ 
\begin{array}{l l}
  1 & :\quad y\\
  * & :\quad \overline{y}\\
\end{array} \right.
\) \hspace{1.7cm}
\(
H_u' = \left\{ 
\begin{array}{l l}
  * & :\quad z\\
  0 & :\quad \overline{z}\\
\end{array} \right.
\) \hspace{0.7cm};
\hspace*{0.7cm}\(
H_u'' = \left\{ 
\begin{array}{l l}
  1 & :\quad x y\\
  \ast & :\quad x \overline{y}\\
  \ast & :\quad \overline{x} z\\
  0 & :\quad \overline{x} \overline{z}\\
\end{array} \right.
\)

Note that the input strategies $H_u$, $H_u'$ need not be consistent for an `$\bowtie$' operation, but they must be in case of an `$\circ$' operation.

\subsection{Definition of \mergeresR}\label{sec:def-mergeR}
Let $\Phi= Q.\phi$ be a QBF with existential variables $X$ and universal variables $U$. A \mergeresR derivation of $L_m$ from $\Phi$ is sequence $\pi = L_1,...,L_m$ of lines where each $L_i= (C_i,\{{H}^{u}_i: u \in U\})$ in which at least one of the following holds for $i \in [m]$:
\begin{itemize}
    \item [a.] \textbf{Axiom}. There exists a clause in $C \in \phi$ such that $C_i$ is the existential subclause of $C$, and for each $u \in U$, $H^u_i$ is the strategy function for $u$ mapping it to the falsifying $u$-literal for $C$ or,
    \item [b.] \textbf{Resolution}. There exist integers $a,b < i$ and an existential pivot $x \in X$ such that $C_i = res(C_a,C_b,x)$ and for each $u \in U$:
    \begin{itemize}
        \item [i.] if $x <_Q u$, then $H^u_i= H^u_b \overset{x}{\bowtie} H^u_a$
        \item [ii.] else if $x >_Q u$, then $H^u_i= H^u_a \circ H^u_b$.
    \end{itemize}
\end{itemize}
$\pi$ is a refutation of $\Phi$ iff $C_m=\bot$. Size of $\pi$ is the number of lines i.e $|\pi|= m$. Observe that similar to \mergeres, proof systems in \mergeresR have only existential literals in the clause part of the lines in a proof.

Let $S$ be a subset of the existential variables $X$ of a false QBF $\mathcal{F}$. We say that a $\mathcal{P}$-refutation $\pi$ of $\mathcal{F}$ (where $\mathcal{P}\in$ \mergeresR) is $S$-regular if for each $x \in S$, there is no leaf-to-root path in $G_{\pi}$ that uses $x$ as pivot more than once. An $X$-regular proof is simply called a regular proof.

\section{Soundness and Completeness of \mergeresR}\label{sec:soundness-completeness}
In this section, we show that each proof system in \mergeresR is sound and refutationally complete for false QBFs. 
\subsection{Soundness}
The soundness proof of \mergeresR follows closely to that of the \mergeres proof system. The following lemma proves the soundness of \mergeresR family of proof systems.
\begin{lemma}\label{lemma_sound}
Let $\mathcal{P} \in$ \mergeresR be any proof system. Let $\pi = L_1, L_2, \dots, L_m$ be a valid $\mathcal{P}$-proof of QBF $\Phi= Q.\phi$. Then, the strategy functions $ \{H^u_m : u\in U\}$ in the conclusion line $L_m$ of
$\pi$, will form a countermodel for $\Phi$. 
\end{lemma}
\begin{proof}
Given $\pi := L_1, . . . , L_m$ be an $\mathcal{P}$-refutation of a QBF $\Phi= Q.\phi$. Let each $L_i = (C_i , \{ H^u_i : u \in U \})$ and $X,U$ are sets of all existential and universal variables in $Q$ respectively. Further, for each $i \in [m]$,

\begin{itemize}
    \item [$\bullet$] let $\alpha_i := \{\overline{l}: l \in C_i \}$ be the smallest assignment falsifying Ci ,
    \item [$\bullet$] let $A_i := \{ \alpha \in  \langle X \rangle : C_i \cap \alpha = \emptyset \}$ be all complete assignments to $X$ consistent with $\alpha_i$,
    \item [$\bullet$] for each $\alpha \in A_i$, let $l^u_i (\alpha) := H^u_i(\alpha)$ and $H_i(\alpha) := \{l^u_i (\alpha): u \in U \} \setminus \{\ast\}$.
\end{itemize}

{\bf Induction statement}:
By induction on $i \in [m]$, we show, for each $\alpha \in A_i$, that the restriction of $\phi$ by $\alpha \cup H_i(\alpha)$ contains the empty clause.

{\bf Proof:}
For the base case $i = 1$, let $\alpha \in A_1$. As $L_1$ is introduced as an axiom, there exists a clause $C \in \phi$ such that $C_1$ is the existential subclause of $C$, and each $H^u_1$ is the function outputting the falsifying $u$-literal for $C$. Hence, for each $u \in U$, $l^u_1 (\alpha)$ is the falsifying $u$-literal for $C$, so $C[\alpha \cup H_1(\alpha)] = \emptyset$.

For the inductive step, let $i \geq 2$ and let $\alpha \in A_i$. The case where $L_i$ is introduced as an axiom is identical to the base case, so we assume that $L_i$ was derived by resolution. Then there exist integers $a, b < i$ and an existential pivot $x \in X$ such that $C_i = res(C_a,C_b, x)$
\begin{itemize}
    \item [(1)] Suppose that $\overline{x} \in \alpha$, each $u \in U$ has to satisfy either:\hspace*{1cm}
        \begin{itemize}
        \item [(i)] $x <_Q u$ and $H^u_i= H^u_b \overset{x}{\bowtie} H^u_a$: In which case, $l^u_i (\alpha) = l^u_a (\alpha)$.
        \item [(ii)] $x >_Q u$ and $H^u_i= H^u_a \circ H^u_b$: In which case, $l^u_i (\alpha) = \{ l^u_a (\alpha) \cup l^u_b (\alpha) \}$.\\
    \end{itemize}
It follows that $l^u_i \neq l^u_a$ only if $l^u_a = \ast$, and hence $H_a(\alpha) \subseteq H_i(\alpha)$. Since $C_a \setminus \{x\} \subseteq C_i$, we have $\alpha \in A_a$, so the restriction of $\phi$ by $\alpha \cup H_i(\alpha)$ contains the empty clause by the inductive hypothesis that $\alpha \cup H_a(\alpha)$ contains the empty clause.\\
    \item [(2)] Suppose that $x \in \alpha$. A similar argument shows that $H_b(\alpha) \subseteq H_i(\alpha)$.
\end{itemize}
\qed
Since $\alpha_m$ is the empty assignment, we have $A_m = \langle X \rangle$ (i.e all complete assignments to $X$). We therefore prove the lemma at the final step $i = m$, as we show that $\{H^u_m: u \in U\}$ is a countermodel for $\Phi$.
\qed
\end{proof}
\subsection{Completeness}
One would notice that a major change of \mergeresR from \mergeres is the usage of `consistency' check instead of `isomorphism' check. Note that the relation between them is as such: isomorphism $\Rightarrow$ consistency but not vice-versa. We use this in our proofs for completeness of \mergeresR.

We show the completeness by first showing that \mergeresM p-simulates \mergeres (Theorem \ref{theorem_3}). Here, \mergeresM is a proof system in \mergeresR which uses merge maps as the representation. Further, we show in Claim \ref{claim_4} that a \mergeresM-proof can be transformed into any \mergeresR-proof in exponential time. Nevertheless, completeness is guaranteed as \mergeres is complete and any \mergeres-proof can be transformed into a \mergeresM-proof which in-turn can be transformed as any \mergeresR-proof.

We will need the following remark from the paper introducing \mergeres \cite{mres_paper}.
\begin{remark}\label{remark-same-function}\cite[Proposition 10]{mres_paper}
Any two isomorphic merge maps compute the same function.
\end{remark}
\begin{theorem}\label{theorem_3}
\mergeresM p-simulates \mergeres.
\end{theorem}
\begin{proof}
Given a QBF $\Phi$ and its \mergeres-proof $\pi=L_1,...,L_m$, where every line $L_i=\{ C_i, \{ M^u_i :u \in U \}\}$. We intend to build an \mergeresM-proof $\Pi=L'_1,...,L'_m$ for $\Phi$, where each $L'_i=\{ C'_i, \{ H^u_i :u \in U \}\}$.

For every line $L_i$ in $\pi$ starting from $i=1$ to $m$, if $L_i$ is an axiom step then directly $C'_i=C_i$ and $H^u_i=M^u_i$ for all $u \in U$. Otherwise, if $L_i$ is an resolution step i.e for some $a,b < i$, $C_i=res(C_a,C_b,x)$; then set $C'_i=C_i$ and for each $u \in U$ if $x <_Q u$ then set $H^u_i= H^u_b \overset{x}{\bowtie} H^u_a$ else set $H^u_i= H^u_a \circ H^u_b$. We see that these are sound steps as resolution in \mergeres can be of the following types:
\begin{itemize}
    \item [(i)] $x >_Q u$ and $M_i = $ select($M^u_a,M^u_b$) ; in this case we set $H^u_i= H^u_a \circ H^u_b$ which holds given the Remark~\ref{remark-same-function} and that isomorphism $\Rightarrow$ consistency.
    \item [(ii)] $x <_Q u$ and $M_i = $ merge($M^u_a,M^u_b,i,x$) ; in this case we set $H^u_i= H^u_b \overset{x}{\bowtie} H^u_a$ which is same as the merge function of \mergeres.
    \item [(iii)] $x <_Q u$ and $M_i = $ select($M^u_a,M^u_b$) ; in this case we set $H^u_i= H^u_b \overset{x}{\bowtie} H^u_a$ which is allowed as \mergeres did the isomorphism test on $M^u_a \text{and} M^u_b$, but we need no such check for $\bowtie$ in \mergeresR. (ref. the note just after Definition~\ref{if_else_def}).
\end{itemize}
In case-(iii) above it remains to note that adding a $\bowtie$ to two isomorphic maps or when one of them is $\ast$, doesn't add any new strategy: it just dilutes the strategy represented by the corresponding merge map. That is, we are adding an if-else condition where both the outcomes are same or one of them is $\ast$.
Hence doesn't affect future consistency checks which may arise in the proof. 
(For further clarity, one is suggested to look at Example \ref{eg_completeness} but it is not needed for the proof).

\noindent 
Therefore, one can clearly see that the proof $\Pi$ constructed in this process is a valid \mergeresM-proof for $\Phi$. Hence this proves the above theorem. \qed
\end{proof}

\begin{claim}\label{claim_4}
Every \mergeresM-proof can be transformed into an \mergeresR-proof for any representation $R$ in exponential time.
\end{claim}
\begin{proof}
Given a QBF $\Phi$ and its \mergeresM-proof $\pi=L_1,...,L_m$, where every line $L_i=\{ C_i, \{ M^u_i :u \in U \}\}$. We intend to build an \mergeresR-proof $\Pi=L'_1,...,L'_m$ for $\Phi$, where each $L'_i=\{ C'_i, \{ H^u_i :u \in U \}\}$.

For every line $L_i$ in $\pi$, we keep the clause part intact while we convert the merge maps into plain functions. Further as $R$ is a complete representation, these functions should have a corresponding representation in $R$; we extensively search for the same. This search terminates at some point owing to $R$ being a complete representation. (This is the place where we used the property that $R$ is a complete representation).
The result is an \mergeresR-proof for $\Phi$.
This process is not polynomial in time but regardless still proves completeness for the family of proof systems \mergeresR. \qed
\end{proof}

Let us consider example~\ref{eg_completeness} (below) which was referred in Theorem~\ref{theorem_3}. This example considers the situation corresponding to the case-(iii) of Theorem~\ref{theorem_3}. That is, two isomorphic merge maps can be combined with an if-else and the resulting strategy will still output the same as input merge maps. Or when one of the input merge map being $\ast$, makes the resulting strategy diluted in the sense that for half the assignments it gives a $\ast$ and for others the same as the non-trivial input merge map.
\begin{example}\label{eg_completeness}
Let $M^u_1 = M^u_2 = 1$ be leaf nodes in \mergeres proof. It implies that corresponding $H^u_1 = 1$ and $H^u_2 =1$ in \mergeresR proof. Now say \mergeres performs a resolution on pivot variable $x$ which is to the left of $u$, resulting in $M^u_3= select(M^u_1,M^u_2)$. Whereas the corresponding \mergeresR rule needs to be a $H^u_3= H^u_1 \overset{x}{\bowtie} H^u_2$ from case(iii) (ref.~Theorem \ref{theorem_3}). That is, $H^u_3$ in function form would be defined as follows:\\
\hspace*{2.5cm}\(
H^u_3 = \left\{ 
\begin{array}{l l}
  1 & :\quad x\\
  1 & :\quad \overline{x}\\
\end{array} \right.
\)\\
Notice how this is just a diluted way of writing the strategy $H^u_3 =1$. Hence when in the next line of \mergeres if a $M^u_4 =1$ which is isomorphic to $M^u_3$ is encountered; the corresponding $H^u_4 =1$ in \mergeresR will still remain to be consistent with $H^u_3$ (though they might seem to be structurally different).

\noindent 
In the same example if $M^u_2 =*$ (i.e. trivial), the strategy $H^u_3$ would have been:\\
\hspace*{2.5cm}\(
H^u_3 = \left\{ 
\begin{array}{l l}
  1 & :\quad x\\
  * & :\quad \overline{x}\\
\end{array} \right.
\)\\

\noindent 
Notice how this is another way of diluting the strategy and is still consistent with $H^u_4=1$. 
\end{example}

So far, we showed that each proof system in \mergeresR is sound and refutationally complete for false QBFs. 
Next, we present an example of \mergeresR allowing few resolution steps which are not allowed in \mergeres. Such examples may be useful for the separation results between the proof systems in \mergeresR and the \mergeres proof system.
\begin{example}\label{eg:motivation}
Consider any proof system $\mathcal{P}$ in \mergeresR which uses some complete $R$ representation for strategies. The following Table \ref{mresP_proof} is a $\mathcal{P}$-refutation of the false QBF : $\exists x \forall u \exists y~(y \lor x \lor u) \land (y \lor \overline{x}) \land (\overline{y} \lor x) \land (\overline{y} \lor \overline{x} \lor \overline{u})$
\begin{table}
\centering
\addtolength{\tabcolsep}{6pt}
\begin{adjustbox}{width=0.75\textwidth}
\scriptsize
\begin{tabular}{c c c c}
\hline
Line & Rule & $C_i$ & $H^u_i$\\ 
\hline
 $L_1$ & axiom & $\{y,x\}$ & 0\\ 
 $L_2$ & axiom & $\{y,\overline{x}\}$ & *\\
 $L_3$ & $res(L_1,L_2,x)$ & $\{y\}$ & $H^u_2 \overset{x}{\bowtie} H^u_1$\\
 $L_4$ & axiom & $\{\overline{y},x\}$ & *\\
 $L_5$ & axiom & $\{\overline{y},\overline{x}\}$ & 1\\
 $L_6$ & $res(L_4,L_5,x)$ & $\{\overline{y}\}$ & $H^u_5 \overset{x}{\bowtie} H^u_4$\\
 $L_7$ & $res(L_3,L_6,y)$ & $\{\}$ & $H^u_3 \circ H^u_6$\\
 \hline
\end{tabular}
\end{adjustbox}
\vspace{0.15cm}
\caption{$\mathcal{P}$-refutation, where $\mathcal{P} \in$ \mergeresR, of the false QBF in Example~\ref{eg:motivation}}
\label{mresP_proof}
\vspace{-0.75cm}
\end{table}

\noindent The strategies $H^u_3$ and $H^u_6$ in function format are as follows:\\
\hspace*{1cm}\(
H^u_3 = \left\{ 
\begin{array}{l l}
  0 & :\quad x=0\\
  * & :\quad x=1\\
\end{array} \right.
\) \hspace{1.5cm}
\(
H^u_6 = \left\{ 
\begin{array}{l l}
  * & :\quad x=0\\
  1 & :\quad x=1\\
\end{array} \right.
\)\\
One can see that these strategies are consistent (but not isomorphic), hence the resolution of $L_3,L_6$ on $y$ is allowed in the $\mathcal{P}$-refutation. But the analogous resolution would be blocked in \mergeres since the corresponding merge maps $M^u_3,M^u_6$ will not be isomorphic.
\end{example}

\section{\qef simulates \mergeresR}\label{sec:qef-simulates-mresr}
In this section, we show that \qef can efficiently simulate any valid refutation from proof system in \mergeresR. Therefore, the stronger proof systems like \qrat (Quantified Resolution Asymmetric Tautologies)~\cite{qrat_paper} and $G$ (Gentzen/Sequent Calculus)~\cite{g_paper} can also simulate the same. 

However, proof systems in \mergeresR can have arbitrary representations. Simulating the same is a nightmare for even the strongest proof systems. But observe that the rules to construct the strategies in any representation are the same as defined in Section~\ref{sec:def-mergeR}. We capture these rules as a new tree structure $T$. That is, given a valid proof $\pi$ of any proof system $\mathcal{P} \in$ \mergeresR, $\pi = (C_1, R_1), (C_2, R_2), \dots, (C_m, R_m)$, we construct a proof $\pi' := (C_1, T_1), (C_2, T_2),\dots , (C_m, T_m)$ in such a way that the representation $T_i$ captures the rules that have been used to construct the strategy $R_i$.

We show that the representation $T$ is also a complete representation for finite functions. Therefore, \mergeresT also belongs to the family \mergeresR. Further, we show that any valid $\mathcal{P}$-proof $\pi$ (where $\mathcal{P} \in$ \mergeresR) can be efficiently converted to a \mergeresT-proof $\pi'$.  
Finally using the ideas from~\cite{Chew2021}, we show that \qef can efficiently simulate \mergeresT. This shows that \qef can simulate any proof systems in \mergeresR. We now proceed and define the proof system \mergeresT.

\subsection{\mergeresT proof system}\label{sec:t-representation}
Given a false QBF $\mathcal{F}$, a \mergeresT proof $\pi$ of $\mathcal{F}$ is a sequence of lines 
$$(C_1,T_1), (C_2, T_2), \dots, (C_m, T_m)$$
where $C_m = \bot$ and each $T_i$ is constructed as follows:
if $(C_i,T_i)$ is an axiom step, then $T_i$ is constructed as in the \mergeres proof system. Otherwise if $(C_i,T_i)$ is a Resolution step on a pivot left of the universal variable in question (i.e if-else step (`$\bowtie$') of \mergeresR), then $T_i$ is constructed, as a merge node is constructed in \mergeres. Further, if $(C_i,T_i)$ is constructed from a resolution step on $(C_j,T_j)$ and $(C_k,T_k)$ with pivot being right of the universal variable in question and both $T_j$ and $T_k$ are consistent (i.e union step (`$\circ$') of \mergeresR), then $T_i$ is constructed by adding a new type of node called the \# node (defined below) with inputs $T_j$ and $T_k$. 

The \# node is defined assuming both its inputs are consistent, and it outputs the result of a union operation on them; more clearly, it's truth table is shown in the Fig.~\ref{fig:subfig1}.
\begin{figure}[h]
\centering
\begin{subfigure}[h]{0.45\textwidth}
\centering
{\renewcommand{\arraystretch}{1.2}
\begin{tabular}{|c|c||c|}
       \hline
A & B & A \# B\\
\hline
1 & 1 & 1\\
0 & 0 & 0\\
$\ast$ & 0/1 & 0/1\\
0/1 & $\ast$ & 0/1\\
$\ast$ & $\ast$ & $\ast$\\
\hline
\end{tabular}}
\caption{Truth table for \# operator.\\(It assumes inputs to be consistent.)}
\label{fig:subfig1}
\end{subfigure}
\hfill
\begin{subfigure}[h]{0.45\textwidth}
\centering
\input{t_example_graph}
\caption{$T^u_{13}$ graph for Example \ref{exp}}
\label{fig:subfig2}
\end{subfigure}
\caption{Truth table of \# operator is shown in Fig.
\ref{fig:subfig1} and its use in \mergeresT depicted by an example QBF in Fig.~\ref{fig:subfig2}}
\label{fig:subfigureExample}
\end{figure}
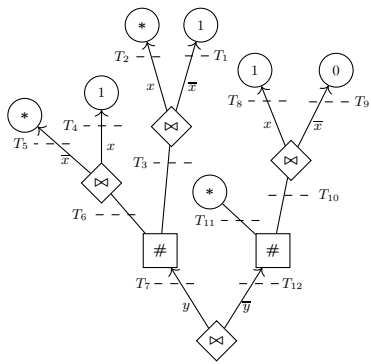
\comment{
\begin{table}[h]
\begin{displaymath}\label{table}
\begin{array}{|c|c||c|}
\hline
a & b & a \# b\\
\hline
1 & 1 & 1\\
0 & 0 & 0\\
* & 0/1 & 0/1\\
0/1 & * & 0/1\\
* & * & *\\
\hline
\end{array}
\caption{Truth table for \# operator. It assumes that both its input are consistent.}
\end{displaymath}
\end{table}}

Note that $A=1,B=0$ and vice-versa cannot happen in a valid \mergeresR proof owing to the definition of union(`$\circ$') which needs the input strategies to be consistent. Therefore, the corresponding rows are omitted from the \# node truth table in Fig~\ref{fig:subfig1}.
Let us illustrate a \mergeresT-proof below for an example QBF. 
\begin{example}\label{exp}
Let $\Phi := \exists x,y, \forall u, \exists a,b ~ (x,\overline{y},\overline{u},a) \land (\overline{x},\overline{y},a) \land (\overline{x},\overline{y},\overline{u},\overline{a}) \land (x, \overline{y},\overline{a}) \land (\overline{x},y,\overline{u},b) \land (x,y,u,b) \land (y,\overline{b})$. The \mergeresT proof of $\Phi$ is shown below in Table \ref{exp_table}:
\begin{table}
\addtolength{\tabcolsep}{6pt}
\begin{adjustbox}{width=1\textwidth}
\scriptsize
\begin{tabular}{c c c c c}
\hline
Line & Rule & $C_i$ & $T^u_i$ & Type of node\\ 
\hline
 $L_1$ & axiom & $\{x,\overline{y},a \}$ & 1 & Leaf\\ 
 $L_2$ & axiom & $\{\overline{x},\overline{y},a \}$ & * & Leaf\\
 $L_3$ & $res(L_1,L_2,x)$ & $\{\overline{y},a \}$ & $T^u_2 \overset{x}{\bowtie} T^u_1$ & if-else\\
 $L_4$ & axiom & $\{\overline{x},\overline{y},\overline{a} \}$ & 1 & Leaf\\
 $L_5$ & axiom & $\{x,\overline{y},\overline{a} \}$ & * & Leaf\\
 $L_6$ & $res(L_5,L_4,x)$ & $\{\overline{y}, \overline{a} \}$ & $T^u_4 \overset{x}{\bowtie} T^u_5$ & if-else\\
 $L_7$ & $res(L_3,L_6,a)$ & $\{\overline{y} \}$ & $T^u_3 \circ T^u_6$ & \#\\
 $L_8$ & axiom & $\{\overline{x},y,b \}$ & 1 & Leaf\\
 $L_9$ & axiom & $\{x,y,b \}$ & 0 & Leaf\\
 $L_{10}$ & $res(L_9,L_8,x)$ & $\{y,b \}$ & $T^u_8 \overset{x}{\bowtie} T^u_9$ & if-else\\
 $L_{11}$ & axiom & $\{y,\overline{b} \}$ & * & Leaf\\
 $L_{12}$ & $res(L_{10},L_{11},b)$ & $\{y \}$ & $T^u_{10} \circ T^u_{11}$ & \#\\
 $L_{13}$ & $res(L_{12},L_7,y)$ & $\{ \}$ & $T^u_7 \overset{y}{\bowtie} T^u_{12}$ & if-else\\
 \hline
\end{tabular}
\end{adjustbox}
\vspace{0.15cm}
\caption{A \mergeresT refutation of the false QBF in Example~\ref{exp}}
\label{exp_table}
\end{table}
\vspace{-0.5cm}
\end{example}

The final $T$-graph of winning strategy for the only universal variable $u$ from Example~\ref{exp} is shown in Figure \ref{fig:subfig2}. One can see that this graph is a hybrid structure of both branching programs and circuits. Since it has both `branching' nodes ($\bowtie$ nodes) and `circuit' nodes (\# nodes). 

Observe that the proposed $T$ representation is complete. That is, any valid finite function can be represented by a $T$ graph. This follows since, merge maps are a subset of $T$-graphs (i.e without \# nodes) which are just branching programs, but known to be complete for all valid functions.
Since $T$ representations are complete, \mergeresT is a member of \mergeresR proof systems. Therefore this is a sound and complete proof system. Also note that \mergeresT is not claimed to be polynomial time verifiable.

\subsection{Conversion of \mergeresR proofs into \mergeresT proofs}
In this section we show how to convert a valid $ \mathcal{P} $-proof $\pi$ into a valid \mergeresT-proof $\pi'$, where $ \mathcal{P}$ be any proof system in \mergeresR. Let $\pi = (C_1, R_1), (C_2, R_2), \dots, (C_m, R_m)$ be a valid $\mathcal{P}$ proof of a QBF $\mathcal{F}$. 
We show how to convert $\pi$ into a valid \mergeresT-proof $\pi' = (C_1, T_1), (C_2, T_2), \dots, (C_m, T_m)$ of the same QBF $\mathcal{F}$. Note that here $T_i$ is not the representation of $R_i$, but $T_i$ is capturing how $R_i$ has been constructed from some hypothesis $R_j,R_k$ with $j,k <i$ using rules from Section \ref{sec:def-mergeR}. For this we do not need to interpret $R_i$'s, but we can extract the required information from the clauses $C_j, C_k$ and $C_i$ of $\pi$. 

It is also useful to note that, during this conversion, one doesn't need to check if the two strategies $R_j,R_k$ are consistent or not. The conversion is smooth and simple as it assumes $\pi$ to be a valid $\mathcal{P}$-proof of $\mathcal{F}$.
We now proceed to give a detailed method for the same.
\begin{theorem}\label{theorem_8}
Any valid $\mathcal{P}$-proof ($\mathcal{P} \in$ \mergeresR) can be converted efficiently into an \mergeresT proof.
\end{theorem}
\begin{proof}
For a false QBF $\mathcal{F}$, proofs of
proof systems belonging to \mergeresR can have arbitrary representations for the strategies computed. However, the rules allowed to construct a strategy $R_i$ using any strategies $R_j$ and $R_k$ (where $j,k <i$) are fixed. They must follow the rules mentioned in Section~\ref{sec:def-mergeR}. \mergeresT proof $\pi'$ captures these rules only. 

To be precise, given a $\mathcal{P}$-proof $\pi$ of $\mathcal{F}$ where $\pi =(C_1,P_1), (C_2,P_2),..., (C_m,P_m)$, we construct \mergeresT-proof $\pi'$ as follows:

From the clause part of the proof $\pi$ i.e $C_1,...,C_m$ (in this sequence) based on what step is being followed (axiom, or resolution where pivot is on left, or resolution where pivot is on right), we build the corresponding $T$-maps as shown in the Figure~\ref{fig:rules}.
\begin{figure}[h]
\centering
\begin{subfigure}[h]{0.15\textwidth}
\centering
\input{fig_1}
\caption{Axiom}
\label{fig2:subfig1}
\end{subfigure}
\hfill
\begin{subfigure}[h]{0.4\textwidth}
\centering
\input{fig_2}
\caption{if-else node \\($res$ when $x$ is left of $u$ in prefix)}
\label{fig2:subfig2}
\end{subfigure}
\begin{subfigure}[h]{0.4\textwidth}
\centering
\input{fig_3}
\caption{\# node (i.e union node) \\($res$ when $x$ is right of $u$ in prefix)}
\label{fig2:subfig3}
\end{subfigure}
\caption{Rules to construct $T$-graphs. In Figure \ref{fig2:subfig1}, $c_u$ is the falsifying strategy of u for the axiom clause $C_i$. In Figure \ref{fig2:subfig2}, $C_i=res(C_j,C_k,x)$ and $x$ is left of $u$ in prefix i.e $T^u_i= T^u_k \overset{x}{\bowtie} T^u_j$. In Figure \ref{fig2:subfig3}, $C_i=res(C_j,C_k,x)$ and $x$ is right of $u$ in prefix i.e $T^u_i= T^u_j \circ T^u_k$. Note that the truth table of the `\# gate' is defined in Figure \ref{fig:subfig1}}
\label{fig:rules}
\end{figure}
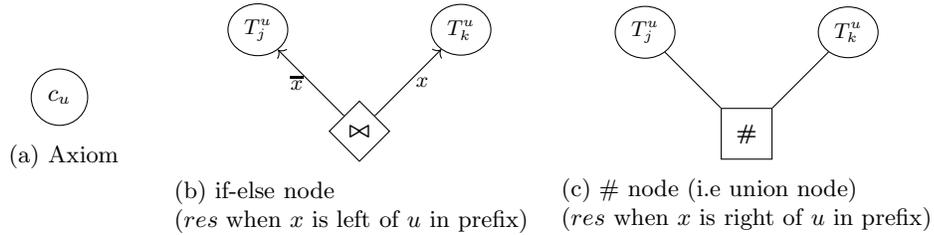

After following this procedure for all lines in $\pi$, the sequence of lines so formed i.e $\pi' = (C_1, T_1), (C_2, T_2), \dots, (C_m, T_m)$ is a valid \mergeresT proof as the clauses $C_1,...,C_m$ are the same as in the original \mergeresR proof hence we know that $C_m$ is definitely $\bot$ and that $T_1,...,T_m$ are built using the same rules as used when building the valid $\mathcal{P}$-proof $\pi$.  Therefore $T_m$ is a countermodel as it is building the same strategy as in $R_m$.
\qed
\end{proof}

Now we proceed to show that \mergeresT proof system can be efficiently simulated by \qef. However before proving the same, observe that \mergeresT efficiently simulates \mergeres proof system: due to Theorem~\ref{theorem_8}, \mergeresT simulates any \mergeresR proof system, and therefore, it also simulates efficiently the \mergeresM $\in$ \mergeresR proof system, which is known to simulate the \mergeres proof system efficiently (Theorem~\ref{theorem_3}). Thus we have the following:
\begin{proposition}\label{prop:mrest-simulates-mres}
\mergeresT efficiently simulates \mergeres.
\end{proposition}
\subsection{\qef simulates \mergeresT}
In this section, we show that \qef efficiently simulates valid \mergeresT refutations. We use the ideas from~\cite[Theorem 1]{Chew2021} which shows how \qef efficiently simulates \mergeres. Let us briefly explain the idea from~\cite[Theorem 1]{Chew2021}: Given an \mergeres-proof $\pi = (C_1, M_1), (C_2, M_2),..., (C_m,M_m)$ of a false QBF $\mathcal{F}$, we know that if $\pi$ is valid then the merge map $M_m$ in the last line gives a winning strategy $S$ for the universal player of $\mathcal{F}$. That is, if we assign values of the universal variables based on $S$, it falsifies $\mathcal{F}$. In~\cite{Chew2021} 
they derived an \qef proof $\pi'$ from $\pi$ efficiently in two phases: in the first phase, they derived $\mathcal{F} \sststile[ss]{\qeff}{} (S \rightarrow \bot)$ using $\pi$. This is equivalent to $\mathcal{F} \sststile[ss]{\qeff}{} \overline{S}$. This first phase was purely propositional. 
Later in the second phase, they used universal reduction to prove $\overline{S} \sststile[ss]{\text{\qef}}{} \bot$. Implying from both phases that $\mathcal{F} \sststile[ss]{\text{\qef}}{} \bot$.
We also use the same tricks for simulating \mergeresT with \qef. Hence, we also simulate the same in two phases. However, in the first phase, they used a double induction in which the second induction depicted how to handle `Select' and `Merge' nodes of \mergeres.
We extend this to \mergeresT by introducing `\# nodes' and giving a method to handle those in the second induction. We now prove this in detail.

\begin{theorem}\label{theorem:simulation}
\qef efficiently simulates \mergeresT.
\end{theorem}
\begin{proof}
{\bf Phase-1}:\\
Given a valid \mergeresT proof $\pi:= (C_1,T_1), (C_2,T_2),..., (C_m,T_m)$ of a false QBF $\mathcal{F}$, 
we create new extension variables for each node in every strategy appearing in the proof. That is, $s^u_{i,t}$ is created for the node $t$ in the strategy $T^u_i$ for the universal variable $u$. 

We define $s^u_{i,t}$ based on whether the corresponding $T^u_i (t)$ is an axiom node, if-else node or \# node as follows:
\[
s^u_{i,t} := \left\{ 
\begin{array}{l l}
  \{1/0/*\} & \quad T^u_{i}(t)=\{1/0/*\}\\
  (y \land s^u_{i,b}) \lor (\overline{y} \land s^u_{i,c}) & \quad T^u_i (t) = T^u_i (b) \overset{y}{\bowtie} T^u_i (c)\\
  s^u_{i,b} \# s^u_{i,c} & \quad T^u_i (t) = T^u_i (b) \circ T^u_i (c)
\end{array} \right.
\]
In the quantifier prefix, we place the newly created variables $s^u_{i,t}$ to the immediate left of $u$ to maintain the soundness of the proof, as strategies for $u$ depends on these variables.

We now prove the outer induction in `Induction \ref{ind_1}' which assumes that `Induction \ref{ind_2}' is valid and hence can derive the clause $C_i$ by assigning local strategies to universal variables through a simple resolution for every line $L_i$.
\begin{induction}\label{ind_1}
Consider the $i^{th}$ line of $\pi$, that is, $(C_i, T_i)$. It is easy for \qef to prove $\bigwedge_{u \in U_i} (u \leftrightarrow s^u_{i,r(u,i)}) \boldsymbol{\rightarrow} C_i$, where $r(u,i)$ is the index of the root node of $T^u_i$. $U_i$ is the subset of $U$ for which $T^u_i$ is non-trivial.
\end{induction}
\textbf{Proof:}\\
\textbf{Base case:} Axiom: Suppose $C_i$ is derived by axiom download of some clause $C \in \mathcal{F}$. If $u$ has a non-trivial strategy, it is because it appears in the clause $C$ and so $u \leftrightarrow s^u_{i,1}$, where $s^u_{i,1} \leftrightarrow c_u$ for $c_u \in \top,\bot$. The constant $c_u$ is correctly chosen to oppose the literal in $C$ so that $C_i$ is just the simplified clause of $C$ replacing all universal $u$ with the corresponding constant $c_u$'s. This is easy for \qef to prove.\\
\comment{
$C=\{x,y,z,\overline{u}\}$, $C_i=\{x,y,z\}$  $\{u \leftrightarrow s^u_{i,1}\}$, $\{s^u_{i,1} \leftrightarrow 1 \}$

$C=\{x,y,z,\overline{u}\}$, $C_i=\{x,y,z\}$  $\{u, \overline{ s^u_{i,1}}\},\{ \overline{x}, s^u_{i,1} \} \{s^u_{i,1}, 0 \} , \{\overline{s^u_{i,1}}, 1 \}$}
\textbf{Inductive step:} Resolution: If $C_j$ is resolved with $C_k$ to get $C_i$ with pivots $\overline{x} \in C_j$ and $x \in C_k$, where $j,k <i$. 
From the induction hypothesis, we have $\bigwedge_{u \in U_j} (u \leftrightarrow s^u_{j,r(u,j)}) \boldsymbol{\rightarrow} C_j$ and $\bigwedge_{u \in U_k} (u \leftrightarrow s^u_{k,r(u,k)}) \boldsymbol{\rightarrow} C_k$. 
Observe that using these clauses, we cannot prove the required statement. However, note that if on the left hand side of theses clauses, one changes the $j$ and $k$ respectively to $i$, then using resolution we can derive $C_i$ on the right hand side. We show in the Induction 2 (below) how to achieve the same. To be precise, from Induction 2 we prove that,
$\bigwedge_{u \in U_i} (u \leftrightarrow s^u_{i,r(u,i)}) \boldsymbol{\rightarrow} C_j$ and $\bigwedge_{u \in U_i} (u \leftrightarrow s^u_{i,r(u,i)}) \boldsymbol{\rightarrow} C_k$ holds.
We then resolve these together to derive $C_i$. This proves Induction \ref{ind_1}.\\
\qed

Now in Induction \ref{ind_2} below, we prove what we claimed before in Induction \ref{ind_1} i.e, given $\bigwedge_{u \in U_j} (u \leftrightarrow s^u_{j,r(u,j)}) \boldsymbol{\rightarrow} C_j$, we show $\bigwedge_{u \in U_i} (u \leftrightarrow s^u_{i,r(u,i)}) \boldsymbol{\rightarrow} C_j$ holds. We proceed by handling each $u \in U_i$ one by one as follows:
\begin{induction}\label{ind_2}
$U_i$ is partitioned into $W$ the set of adjusted variables and $V$ the set of variables yet to be adjusted. For every such $V,W$, the following holds:

$(\bigwedge_{v\in V\cap U_j}(v \leftrightarrow s^v_{j,r(v,j)})) \land (\bigwedge_{w\in W}(w \leftrightarrow s^w_{i,r(w,i)})) \boldsymbol{\rightarrow} C_j$\\
Recall that $U_i$ is the subset of $U$ for which $T^u_i$ is non-trivial.
\end{induction}
\textbf{Proof:}\\
\textbf{Base case:} Initially $W$ is empty and as strategies cannot go back to be trivial $U_j \subseteq U_i$. Hence the statement to prove is exactly the statement given above in the hypothesis. Therefore, base case is trivially true.\\
\textbf{Inductive step:}\\Starting with $(\bigwedge_{v\in V\cap U_j}(v \leftrightarrow s^v_{j,r(v,j)})) \land (\bigwedge_{w\in W}(w \leftrightarrow s^w_{i,r(w,i)})) \boldsymbol{\rightarrow} C_j$.\\
We pick a $u \in V$ to adjust into $i$-terms, i.e we show the following:\\
$(u \leftrightarrow s^u_{i,r(u,i)}) \land (\bigwedge_{v\in \{V\cap U_j\} \setminus \{u\}}(v \leftrightarrow s^v_{j,r(v,j)})) \land (\bigwedge_{w\in W}(w \leftrightarrow s^w_{i,r(w,i)})) \boldsymbol{\rightarrow} C_j$.\\
We have three cases based on the rule used to derive the line $L_i=(C_i,T_i)$:
\begin{itemize}
    \item [i]  $T^u_j=*$
    \item [ii] $T^u_j \neq *$, $T^u_i = T^u_j \overset{x}{\bowtie} T^u_k$
    \item [iii] $T^u_j \neq *$, $T^u_i= T^u_j \circ T^u_k$
\end{itemize}
In case (i) we can easily adjust the universal variable $u$. That is, we can simply add the following: $(u \leftrightarrow s^u_{i,r(u,i)})$. This is sound because the clause $(u \leftrightarrow s^u_{j,r(u,j)})$ has never appeared before in the left hand side of the hypothesis but still we were able to derive $C_j$. Therefore, adding $(u \leftrightarrow s^u_{i,r(u,i)})$ to the left hand side of the hypothesis, will still be able to derive $C_j$.\\
In case (ii) we prove inductively that for each node $t$ in $T^u_j, (s^u_{i,t} \leftrightarrow s^u_{j,t})$ holds. This is true for all leaf and intermediate nodes of $T^u_j$ as we are only going to connect two $T$-graphs( i.e $T^u_j, T^u_k$) by an extra if-else node in $T^u_i$, i.e. all nodes of $T^u_j$ are present in $T^u_i$.
Hence eventually at the root of $T^u_j$, we will have $s^u_{i,r(u,j)} \leftrightarrow s^u_{j,r(u,j)}$. 
However to prove the induction statement (Induction \ref{ind_2}), we need to show this relation between roots of $T^u_i$ and $T^u_j$ i.e, $s^u_{i,r(u,i)} \leftrightarrow s^u_{j,r(u,j)}$. For this we use the definition of merging that $x \boldsymbol{\rightarrow} (s^u_{i,r(u,i)} \leftrightarrow s^u_{i,r(u,j)})$ and so we have $(s^u_{i,r(u,i)} \leftrightarrow s^u_{i,r(u,j)}) \lor \overline{x}$. We almost got the relation we needed but only $\overline{x}$ is the extra literal.
But note that $\overline{x}$ is already $\in C_j$. So, the $\overline{x}$ is absorbed by the $C_j$ in right hand side of the implication.\\
In case (iii) using a similar induction as used in case (ii), we can derive $s^u_{i,r(u,j)} \leftrightarrow s^u_{j,r(u,j)}$: because we are not deleting any strategies just adding an \# gate.
By the definition of the $\#$ gate, $s^u_{i,r(u,i)} \neq s^u_{i,r(u,j)}$ only when $s^u_{i,r(u,j)} \leftrightarrow *$, 
in which case it is directly case-(i) above. That is, we can directly add $u \leftrightarrow s^u_{i,r(u,i)}$ to the given hypothesis and we are done. In the other case when $s^u_{i,r(u,i)} = s^u_{i,r(u,j)}$, we can simply add $s^u_{i,r(u,i)} \leftrightarrow s^u_{i,r(u,j)}$ which directly proves the induction step.\vspace{0.3cm}\qed
\noindent \textbf{Phase-2:}\\
At this point, from the Induction 1, we have derived:
$$\mathcal{F} \sststile[ss]{\qeff}{} (\bigwedge_{u \in U_m} (u \leftrightarrow s^u_{m,r(u,m)}) \boldsymbol{\rightarrow} \bot)$$
In other words, we have derived the winning strategy (say $S$) for the universal player in the QBF $\mathcal{F}$ i.e $\mathcal{F} \sststile[ss]{\qeff}{} (S \rightarrow \bot)$. Equivalently, $\mathcal{F} \sststile[ss]{\qeff}{} \overline{S}$. 
Also, observe that so far we are only in the propositional world. Using the ideas from~\cite{Chew2021}, now from the negation of the strategies for the universal player (i.e., $\overline{S}$), we can easily derive the empty clause using the universal reduction steps. 

\noindent That is, we have the following $ \overline{S} := \bigvee_{i=1}^n (u_i \oplus  s^{u_i}_{m,r(u_i,m)})$, where $U_m = \{u_1, u_2,\dots, u_n \}$ in this order in the prefix. 
Observe the following property for some $k=1$ to $k=n$ in this order:
$$\bigvee_{i=1}^{n-k+1} (u_i \oplus  s^{u_i}_{m,r(u_i,m)})$$
From the above formula, just pull out the last term and we have:
$$\bigvee_{i=1}^{n-k} (u_i \oplus  s^{u_i}_{m,r(u_i,m)}) \vee (u_{n-k+1} \oplus s^{u_{n-k+1}}_{m,r(u_{n-k+1},m)} )$$
Performing the universal reduction step on $u_{n-k+1}$ is the same as:\\
$$\bigvee_{i=1}^{n-k} (u_i \oplus  s^{u_i}_{m,r(u_i,m)}) \vee (0 \oplus s^{u_{n-k+1}}_{m,r(u_{n-k+1},m)} ) \bigwedge
\bigvee_{i=1}^{n-k} (u_i \oplus  s^{u_i}_{m,r(u_i,m)}) \vee (1 \oplus s^{u_{n-k+1}}_{m,r(u_{n-k+1},m)} )$$
Which is same as:
$$\bigvee_{i=1}^{n-k} (u_i \oplus  s^{u_i}_{m,r(u_i,m)}) \vee (s^{u_{n-k+1}}_{m,r(u_{n-k+1},m)}) \bigwedge
\bigvee_{i=1}^{n-k} (u_i \oplus  s^{u_i}_{m,r(u_i,m)}) \vee (\overline{s^{u_{n-k+1}}_{m,r(u_{n-k+1},m)}})$$
Note that we can perform universal reduction on $u_{n-k+1}$ as the only existential new variable appearing in the clause is to the left of it in the prefix.

We can resolve these two to get the following:\\
$$\bigvee_{i=1}^{n-k} (u_i \oplus  s^{u_i}_{m,r(u_i,m)})$$
Note that we used the following rule above: $0 \oplus x$ is $x$ and $1 \oplus x$ is $\bar{x}$.
We continue reducing all $u_i$'s to derive the $\bot$ at the end. 

The proof of Theorem 9 can be concluded by combining the results of Phase-1 and Phase-2 i.e, $\mathcal{F} \sststile[ss]{\qeff}{} \overline{S}$ and $\overline{S} \sststile[ss]{\text{\qef}}{} \bot$. Implying from both phases that $\mathcal{F} \sststile[ss]{\text{\qef}}{} \bot$. 

\qed
\end{proof}
Note that given a valid \mergeresT proof, by the simulation in Theorem~\ref{theorem:simulation} one can obtain a valid \qef proof. But the validness of the resultant \qef proof cannot be used to determine if the initial \mergeresT proof was valid or not. That is, an invalid \mergeresT proof may result into a valid \qef proof. Therefore as claimed before, these \mergeresR proof systems are not polynomial time verifiable even after being simulated by the powerful \qef proof system.
From the above discussions, Theorem~\ref{theorem_8} and Theorem~\ref{theorem:simulation} imply the following:
\begin{theorem}
\qef efficiently simulates valid refutations from proof systems in \mergeresR.
\end{theorem}


\section{Lower Bound for Regular \mergeresR}\label{sec:lower-bound-regular}
In this section, we lift the lower bound of Completion Formulas (\crn, Section \ref{CR_formula}) for Regular-\mergeres~\cite[Theorem 9]{BeyersdorffBMPS20}, to Regular-\mergeresR. We state the \crn formulas once again for ease of reference.

\crn = $\underset{i,j \in [n]}{\exists} x_{ij}$, $\forall z$, $\underset{i \in [n]}{\exists} a_{i}$, $\underset{j \in [n]}{\exists} b_{j}$
$\Big( \underset{i,j \in [n]}{\bigwedge} ( A_{ij} \wedge B_{ij} ) \Big) \wedge L_A \wedge L_B$\\
where,

$A_{ij}= x_{ij} \lor z \lor a_{i}$ \hspace{1.8cm} $B_{ij}= \overline{x_{ij}} \lor \overline{z} \lor b_{j}$

$L_A= \overline{a_1} \lor \dots \lor \overline{a_n}$ \hspace{1.7cm} $L_B= \overline{b_1} \lor \dots \lor \overline{b_n}$ \vspace{0.2cm}

The lower bound follows from a stronger result that we prove below in Theorem \ref{theorem_lb} that any (A $\cup$ B)-regular refutation of \crn in any proof system belonging to \mergeresR must have size $2^{{O}(n)}$. 
We use the ideas from \cite{BeyersdorffBMPS20} to prove the lower bound.
We try to maintain the same notations wherever possible for simplicity. 

Before presenting the lower bound proof in detail, we present the basic idea for the same. The proof setup is depicted in Figure \ref{fig:proof_diag}.
As every clause in \crn has a variable from the set $A \cup B$, but the refutation should derive a $\bot$ at the final line; there must be a `section' of the proof (See shaded region $S'$ in Fig \ref{fig:proof_diag}) which only has $X$ variables in all its clauses. This section also includes the final line. The set of clauses at the `border' (See the bold line $S$ in Fig \ref{fig:proof_diag}) of this section of the proof is shown to be wide (in terms of number of literals) in Lemma \ref{lemma_6}. Using this and the argument that the conjunction of clauses in $S$ itself forms a false CNF formula, we show in Theorem \ref{theorem_lb} that the number of clauses in $S$ is large (exponential in $n$). This directly implies that the size of the \mergeresR-proof is also large.

To establish the width bound, we note that the pivots which are used while deriving clauses in $S$ are variables from $A \cup B$ and that they are all to the right of $z$. Meaning that the corresponding resolutions must all be union steps i.e the incoming strategies must be consistent (not isomorphic as is the case in \mergeres).
This especially makes it difficult to directly lift the lower bound proof of \mergeres from \cite{BeyersdorffBMPS20}.
However we successfully overcome this issue in Claim \ref{claim_7} by arguing how $L_A,L_B$ are the only clauses with trivial strategies and how any other clause which resolves with these will mask this trivial-ness with its own definitive strategy.
Further, by analysing what axiom clauses cannot be used in the derivation of the clauses in $S$, we show that many variables cannot be resolved before these lines. Hence, these variables will still be present in the clause $\in S$, making it wide. We now clearly state and prove the theorem for the lower bound result.

\begin{figure}[h]
\centering
\input{lower_bound}
\caption{\textit{Lower bound proof illustration.} Given any $\mathcal{P} \in$ \mergeresR, a \crn formula and it's $\mathcal{P}$-proof $\Pi$, this figure shows the graph $G_\Pi$. Claim \ref{claim_7} illustrates that $x_{ij} \notin var(H^z_2)$ for $i \in [n-1], j \in [n]$ . Claim \ref{claim_8} illustrates that $|vars(C_2)| \geq n-1$. Lemma \ref{lemma_6} shows that $|vars(C)| \geq n-1$. Theorem \ref{theorem_lb} proves that $|S| \geq 2^{n-1} $}
\label{fig:proof_diag}
\vspace{-0.3cm}
\end{figure}
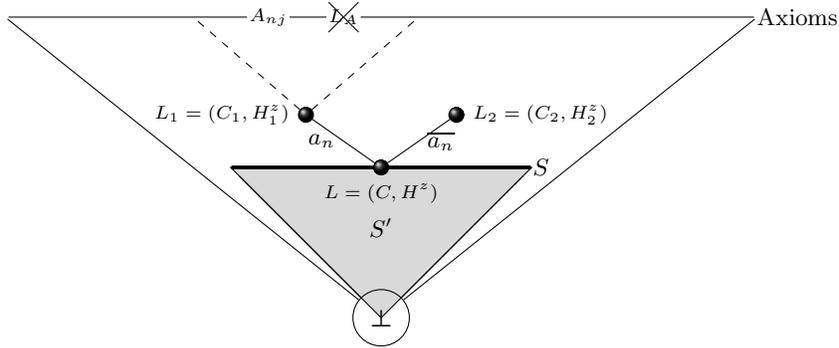
\begin{theorem}\label{theorem_lb}
Every ($A \cup B$)-regular refutation of \crn in any proof system belonging to \mergeresR has size $2^{\Omega(n)}$.
\end{theorem}
\begin{proof}
For $\mathcal{P} \in$ \mergeresR, 
let $\Pi$ be a $\mathcal{P}$-refutation of \crn (for $n > 2$).
Let the set of lines $S, S'$ be defined as follows:
\begin{description}
\item[$S'$:] This set consists of all the lines $L = (C,H^z)$ from $\Pi$ such that $vars(C) \cap \{ A \cup B \} = \emptyset$ and there exists a path from $L$ to $\bot$ in $G_\Pi$ consisting of lines only from $S'$. 
\item[$S$:] This set contains all the lines $L \in S'$ such that $L=Res(L_1,L_2,v)$ where $L_1,L_2 \notin S'$. Observe that the pivot variable $v$ must belong to $ \{ A \cup B \}$.\vspace{0.2cm}
\end{description}

\noindent Let $F = \underset{(C,H^z) \in S}{\bigwedge}C$. Note that $F$ is a false CNF formula because there exists a sub-derivation $\widehat{\Pi} = \{C |\exists L = (C,H^z) \in S' \}$ which derives a $\bot$ given $F$. The variables in $F$ are only $x_{ij}$'s where $i,j \in [n]$, therefore it consists of $n^2$ variables. In Lemma \ref{lemma_6} we prove that each clause in $F$ has width $\geq n-1$. That is each clause can be falsified by setting atleast n-1 variables to 0. Hence the number of complete assignments of $X$ that can falsify a clause $\in F$ will be at most $2^{ n^2 -(n-1)}$. Since $F$ is a false CNF formula, all assignments to $X$ should falsify some clause of $F$. Therefore, the number of clauses in $F$ should be $\geq 2^{n-1}$. This implies that the number of lines in $S$ is at least $2^{n-1}$. Therefore, the number of lines in $\Pi$ must also be exponential in $n$.
\qed
\end{proof}

Now it remains to prove Lemma~\ref{lemma_6} which we prove below.
\begin{lemma}[\cite{BeyersdorffBMPS20}]\label{lemma_6}
For all $L = (C,H^z) \in S$, width($C$) $\geq n-1$.
\end{lemma}
\begin{proof}
Observe that $L$ is not an axiom as all axioms of \crn have a variable from $A \cup B$ and so they cannot belong to $S$. So, let $L = res(L_1,L_2, v)$ where $L_1,L_2 \notin S'$. Since two lines not belonging in $S'$ resolve to make the resultant $\in S'$, the pivot (i.e $v$) should be from $A \cup B$. Assume $v \in A$, a similar argument can be made when $v \in B$. Without loss of generality, assume that $v = a_n$\footnote{Note that here $a_n$ is used only for ease in dividing the set $A$ into partitions. nowhere in the proof we use the fact that $a_n$ is the last variable in $A$. Hence it is indeed w.l.o.g}; and $a_n \in C_1$ and $\overline{a_n} \in C_2$.

Since $\Pi$ is ($A \cup B$)-regular, $a_n$ does not occur as a pivot in the sub-derivation $\Pi_{L_1}$. It implies that the axiom clause $L_A$ cannot be used in deriving $L_1$, because otherwise $C_1$ will have both $a_n$ \& $\overline{a_n}$ which makes it a tautology. That implies, axioms with other positive literals $a_i$'s cannot be used in $\Pi_{L_1}$ as the negated literals $\overline{a_i}$'s are only available in $L_A$ which in-turn cannot be used in $\Pi_{L_1}$. Positive literals of $a_i$'s only $\in A_{ij}$ for all $j \in [n]$. Hence, axioms $A_{ij}$ for $i \in [n-1] , j \in [n]$ also cannot be used in deriving the line $L_1$. Now, we know $x_{ij}$ only occur in $A_{ij}$; so $H^z_1$ has no $x_{ij}$ variable for $i \in [n-1] , j \in [n]$. Also, $H^z_1$ is not a trivial strategy as some $A_{nj}$ for $j \in [n]$ has been used because $a_n \in C_1$.

Since the pivot $a_n$ at the resolution step obtaining line $L$ is to the right of $z$,by the rules of \mergeresR, $H^z_1$ and $H^z_2$ are consistent.
In Claim \ref{claim_7}, we prove that even though \mergeresR only insists on consistency, it still holds that for each $i \in [n-1]$, and each $j \in [n]$, $x_{ij} \notin$ var($H^z_2$). Using this result we prove in Claim \ref{claim_8} below, that $C_2$ will have at least $n-1$ variables (including $\overline{a_n}$). Therefore, at least $n-2$ variables from $C_2$ belong in $C$.

Also, observe that $x_{nj} \in C_1$ for some $j \in [n]$: Since some clause $A_{nj}$ for $j \in [n]$ was used in $\Pi_{L_1}$, the literal $x_{nj}$ is introduced into the proof and resolution of $x_{nj}$ is not possible before $L_1$. This is because, the clause $B_{nj}$ needed to resolve it, brings with it literal $b_j$ which needs to be resolved before $L_1$ (as $L_1$ cannot have any $A \cup B$ literals other than $a_n$). To resolve this $b_j$, one needs to introduce the clause $L_B$, but $L_B$ brings all $\overline{b}$'s into the resultant which cannot be further resolved as the $B$-clauses needed for the same do not have consistent strategies anymore. That is, because of the use of $A_{nj}$ the resolvent has a $0$ strategy for some assignment to $X$ variables, but $B$-clauses have a constant strategy of $1$ hence these strategies will not be consistent to resolve further.


Hence, we know $x_{nj} \in C_1$ for some $j \in [n]$. It implies that $x_{nj} \in C$ as well. This $x_{nj}$ cannot $\in C_2$ as the corresponding axiom clause needed for the same has $a_n$ in it, which would make $C_2$ a tautology.
Using the three results above, we can derive that width($C$) $\geq n-1$.
\qed
\end{proof}

\begin{claim}\label{claim_7}
For $i \in [n-1]$, and each $j \in [n]$, $x_{ij} \notin$ var($H^z_2$).
\end{claim}
\begin{proof}
At the point of use of this claim in the proof of Lemma \ref{lemma_6}, we definitely know that for $i \in [n-1]$ \& $j \in [n]$; $x_{ij} \notin H^z_1$. That is, if $f_1$ is the function representing the strategy $H^z_1$, 
then for any assignment $\sigma$ of $x_{nj}$'s and $i \in [n-1], j \in [n]$, it implies that: 
\begin{align}
f_1(\sigma,x_{ij}=0) = f_1(\sigma,x_{ij}=1)
\end{align}

Let $f_2$ be the function representing the strategy $H^z_2$. Since $a_n$ is to the right of $z$, we know that $H^z_1$ and $H^z_2$ are consistent, i.e for any assignment $\sigma'$ (an extension of $\sigma$) and for $i \in [n-1], j \in [n]$, it implies that:
\begin{align}
f_2(\sigma',x_{ij}=0) \simeq f_1(\sigma',x_{ij}=0) \label{eq:2}\\
f_2(\sigma',x_{ij}=1) \simeq f_1(\sigma',x_{ij}=1) \label{eq:3}
\end{align}
Only remaining question is if $f_2(\sigma',x_{ij}=0) = f_2(\sigma',x_{ij}=1)$? Observe that if this equality holds, then $f_2$ will be independent of $x_{ij}$'s, which implies that $x_{ij} \notin H^z_2$ for $i \in [n-1],j \in [n]$. Now, we are heading towards proving the equality holds. \\
Note that if none of the terms in equation \ref{eq:2} and equation \ref{eq:3} give a `$\ast$' for any assignment of $X$, the equality in question definitely holds. So, now we prove that none of them can give a `$\ast$' for any given assignment.

The only axiom clauses of \crn with trivial strategies are $L_A, L_B$ and these axioms only contain variables of $A \cup B$, which are all to the right of $z$. Hence if any other clause is to be resolved with these clauses, the pivot has to be in $A \cup B$ i.e. a union step needs to be performed. At this point the trivial-ness of $L_A$ (or $L_B$) is masked and does not show up in the final strategy of the resultant line; this is because union of any strategy with a trivial strategy will be the strategy itself.
The only case by which a `$\ast$' can be in the resulting strategy is if $L_A$ is resolved with $L_B$, which can clearly not happen as they have no common variable. 

Since $C_1,C_2$ are definitely not the axiom clauses $L_A$ (or $L_B$), using the above argument it is simply not possible for the functions $f_1$ (or $f_2$) to output a `$\ast$' for any input assignment provided. 
This means the equality in question above holds; meaning that $H^z_2$ also doesn't depend on $x_{ij}$'s when $i \in [n-1], j\in [n]$ i.e $x_{ij} \notin vars(H^z_2)$. \qed
\end{proof}
Now we prove claim \ref{claim_8} which was used in Lemma \ref{lemma_6}.
\begin{claim}[\cite{BeyersdorffBMPS20}]\label{claim_8}
Either for all $i \in [n-1]$, $C_2$ has a variable of the form $x_{i*}$, or for all $j \in [n]$, $C_2$ has a variable of the form $x_{*j}$
\end{claim}
\begin{proof}
At this point in the proof of Lemma~\ref{lemma_6}, we definitely know that $\overline{a_n} \in C_2$, and for all $i \in [n-1]$, for all $j \in [n]$, $x_{ij} \notin$ var($H^z_2$). We prove this claim by contradiction. Suppose the claim is wrong i.e, there exists some $u \in [n-1]$ where for all $l \in [n]$ $x_{ul} \notin var(C_2)$ and some $v \in [n]$ where for all $k \in [n]$ $x_{kv} \notin var(C_2)$.

Let $\rho$ be the minimum partial assignment falsifying $C_2$. Then we know that :
\begin{itemize}
\item [$\triangleright$] $\rho$ sets $a_n = 1$, leaves all other variables in $A \cup B$ unset, since they $\notin C_2$.
\item [$\triangleright$] $\rho$ does not set any $x_{ul}$ or $x_{kv}$, since by our assumptions they all are not in $C_2$.
\end{itemize}
Now, extend $\rho$ to assignment $\alpha$ by setting:
\begin{itemize}
\item [$\triangleright$] $a_u=b_v=0$ and rest all unset variables from $A \cup B$ to $1$. 
\item [$\triangleright$] Also except $x_{uv}$, set $x_{u*}=1$ and $x_{*v}=0$.
\end{itemize}
Observe that the assignment $\alpha$ satisfies all axiom clauses except $A_{uv}$ and $B_{uv}$ and does not falsify any axiom.\\
Now extend $\alpha$ to $\alpha_0$ and $\alpha_1$ by setting $x_{uv}=0$ and $1$ respectively.
 
The extension $\alpha_0$ satisfies one more axiom i.e. $B_{uv}$; similarly $\alpha_1$ satisfies one more axiom i.e. $A_{uv}$. Note that they still do not falsify the remaining axiom. That is, $\alpha_0$ does not falsify $A_{uv}$ and similarly, $\alpha_1$ does not falsify $B_{uv}$.

$\alpha_0$ and $\alpha_1$ agree everywhere except on $x_{ij}$ , and since $x_{ij} \notin$ var($H^z_2$), it follows that $H^z_2(\alpha_0) = H^z_2(\alpha_1)$, say this value is equal $d$.

From the proved Induction in Lemma~\ref{lemma_sound}, the partial strategy of universal player at every line combined with the extension of the existential assignment falsifying it's clause part, should falsify some axiom of the QBF. Also, $\alpha_0$ and $\alpha_1$ falsify $C_2$, since they extend $\rho$. Hence, it is a contradiction that ($\alpha_{\overline{d}},d$) satisfies all axioms. Therefore, the claim needs to be true.






\qed 
\end{proof}
From the above discussions and due to Theorem~\ref{theorem_lb}, we have the following:
\begin{theorem}\label{thm:reg-meresr-lower}
Every \mergeresR-regular refutation of \crn has size $2^{\Omega(n)}$.
\end{theorem}

\section{Conclusion and Future work}\label{sec:conclusion}
\mergeres proof system introduced recently in \cite{mres_paper} builds strategies into proofs for false QBFs. We extend this proof system to a new family of sound, refutationally complete but not polynomial time verifiable proof systems \mergeresR.
For each complete representation $R$, we have a proof system in \mergeresR. 
We also define a complete representation $T$, and it's proof system \mergeresT belonging to \mergeresR. 
We show how this \mergeresT proof system efficiently simulates the before-mentioned \mergeres proof system. We also prove that \qef can simulate every valid refutation from proof systems belonging to \mergeresR. Further, we establish a lower bound of Completion Formulas (\crn) for every regular-proof system in \mergeresR. Refer Fig~\ref{fig:systems_diag} for the resulting landscape of QBF-proof systems with efficient simulations.

Still several open problems remain in the scope of this paper. 
We would like to end our discussions by pointing out a few of them.

The simulation relation between proof systems in \mergeresR and \mergeres is still open. Since proof systems in \mergeresR uses strong consistency checking rules as compared to the isomorphism rule in \mergeres, we believe that there exists a family of QBFs which are easy for proof systems in \mergeresR but hard for \mergeres. For the motivation of the same refer Example~\ref{eg:motivation}. It presents the resolution steps forbidden in \mergeres but allowed in \mergeresR. 

Another important open problem, is to establish a lower bound for proof systems in \mergeresR. Note that whether KBKF-lq formulas from~\cite{BalabanovWJ14}, is hard or easy for proof systems in \mergeresR is still open. These formulas have been shown to be hard for the \mergeres proof system in \cite{BeyersdorffBMPS20}. 

\mergeres proof system is inspired from the \lqrc proof system. It allows some forbidden resolution steps of \lqrc. It has already been shown that \mergeres efficiently simulates the reduction-less \lqrc proof system~\cite{mres_paper}. However, it is still open whether \mergeres and \lqrc are incomparable, or if one can simulate the other. 

\section{Acknowledgements}
We would like to thank Gaurav Sood for many important discussions, comments and suggestions regarding this paper.

\bibliographystyle{plain}
\bibliography{my-ref}

\end{document}

%% file: proof_systems.tex
\begin{tikzpicture}
\node[calcn,draw,minimum width = 0.75cm, minimum height = 0.75cm] (a) at (-0.7,7.2) {$G$};
\node[calcn,draw,minimum width = 0.75cm, minimum height = 0.75cm] (b) at (-0.7,6.1) {\qrat};

\node[calcn,draw,minimum width = 0.75cm, minimum height = 0.75cm] (c) at (-0.7,5.05) {\qef};
\node[rectangle,draw,minimum width = 0.75cm, minimum height = 0.75cm, very thick] (d) at (-0.7,3.5) {\mergeresT};
\node[rectangle,draw,minimum width = 0.75cm, minimum height = 0.75cm, densely dotted, line width=2pt] (e) at (-0.7,2.25) {\mergeresR};
\node[calcn,draw,minimum width = 0.75cm, minimum height = 0.75cm] (g) at (2.35,3) {\qff+$\forall$red};

\draw[black, dashed] (7,1.7) -- (-7.25,1.7)node[pos=.9, below] {Known lower bound};
\draw[black, dashed] (2.25,5.6) -- (-7.25,5.6)node[pos=.82, below] {Known strategy extraction};

\node[calcn,draw,minimum width = 0.75cm, minimum height = 0.75cm] (f) at (-3.2,1) {\mergeres};
\node[rectangle,draw,minimum width = 0.75cm, minimum height = 0.75cm, densely dotted,line width=2pt] (h) at (-0.7,0.5) {Regular-\mergeresR};
\node[calcn,draw,minimum width = 0.75cm, minimum height = 0.75cm] (i) at (-5.2,0.5) {\irmc};
\node[calcn,draw,minimum width = 0.75cm, minimum height = 0.75cm] (j) at (-5.2,-0.8) {\irc};
\node[calcn,draw,minimum width = 0.75cm, minimum height = 0.75cm] (k) at (-5.2,-2) {\ecalculus};
\node[calcn,draw,minimum width = 0.75cm, minimum height = 0.75cm] (l) at (-0.7,-1.2) {\lqrc};

\node[rectangle,draw,minimum width = 0.75cm, minimum height = 0.75cm] (m) at (5.2,1) {\sqcp};
\node[calcn,draw,minimum width = 0.75cm, minimum height = 0.75cm] (n) at (5.2,-0.35) {\qcp}; 
\node[calcn,draw,minimum width = 0.75cm, minimum height = 0.75cm] (o) at (2,-1.8) {\qurc};
\node[calcn,draw,minimum width = 0.75cm, minimum height = 0.75cm] (q) at (2,1.1) {\lquprc};
\node[calcn,draw,minimum width = 0.75cm, minimum height = 0.75cm] (r) at (2,-0.3) {\lqurc};

\node[calcn,draw,minimum width = 0.75cm, minimum height = 0.75cm] (p) at (-0.7,-2.6) {\qrc};

\draw[black, -] (a) -- (b);
\draw[black, -] (b) -- (c);
\draw[black, -, ultra thick] (c) -- (d);
\draw[black, -, ultra thick] (d) -- (e);
\draw[black, -] (c.360) -- (g);
\draw[black, -, ultra thick] (e) -- (h);
\draw[black, -, shorten <= -2pt] (c.200) -- (i.35);
\draw[black, -] (i) -- (j);
\draw[black, -] (j) -- (k);
\draw[black, -, shorten >= -1pt, shorten <= -2pt] (i) -- (l);
\draw[black, -] (l) -- (p);
\draw[black, -, shorten >= -2pt] (g.275) -- (n.155);
\draw[black, -] (m) -- (n);
\draw[black, -, shorten <= -2pt] (n) -- (o);
\draw[black, -] (o) -- (p);
\draw[black, -] (j) -- (p);
\draw[black, -] (c.210) -- (f.110);
\draw[black, -] (c.335) -- (q);
\draw[black, -] (q) -- (r);
\draw[black, -] (r) -- (o);
\draw[black, -, ultra thick] (d.200) -- (f);
\draw[black, -, shorten <= -1pt, shorten >= -1pt] (r) -- (l);
\draw[gray,semithick,rounded corners=2mm] (2.25, 8) rectangle (9,3.8);
\node[] at (3,7.65) {Nodes:};
\node [calcn,minimum height=.65cm, minimum width=.5cm, inner sep=.1cm, font=\small, label={[align=left]0: Polynomial\\ time verifiable}] at (2.75, 6.9) {\textnormal{\hspace{1.5mm}A\hspace*{1.5mm}}}; 
\node [rectangle, draw, minimum width = 0.5cm, minimum height = 0.6cm, font=\small,label={[align=left]0: Not polynomial\\ time verifiable}] at (5.85,6.9) {\textnormal{A}}; 

\draw[densely dotted, black, thick] (2.5, 5.9) -- (3,5.9)node[pos=.7, right, label={[align=left]0: Family of\\proof systems}] {}; 
\draw[-, black] (5.6, 5.9) -- (6.1,5.9)node[pos=.7, right, label={[align=left]0: Single\\proof system}] {}; 

\node[] at (3,5.15) {Edges:};
\draw[-, black] (2.6, 4.75) -- (3.1,4.75)node[pos=1, right, label={[align=left]0: p-simulation result}] {};
\node[] at (5.55,4.2) {nodes/edges drawn in bold are new results.};
\end{tikzpicture}

%% file: t_example_graph.tex

\begin{tikzpicture}[scale=0.6, transform shape]

\node[diamond,draw,minimum width = 0.75cm, minimum height = 0.75cm] (a) at (5,-0.5) {{$\boldsymbol{\bowtie}$}};
\node[rectangle,draw,minimum width = 0.75cm, minimum height = 0.75cm] (b) at (3.75,1.5) {{\textbf{\#}}};
\node[rectangle,draw,minimum width = 0.75cm, minimum height = 0.75cm] (c) at (6.25,1.5) {{$\boldsymbol{\#}$}};
\node[diamond,draw,minimum width = 0.75cm, minimum height = 0.75cm] (d) at (2.45,3) {{$\boldsymbol{\bowtie}$}};
\node[diamond,draw,minimum width = 0.75cm, minimum height = 0.75cm] (e) at (4,4.25) {{$\boldsymbol{\bowtie}$}};
\node[ellipse,draw,minimum width = 0.75cm, minimum height = 0.75cm] (f) at (4.85,2.8) {{$\boldsymbol{\ast}$}};
\node[diamond,draw,minimum width = 0.75cm, minimum height = 0.75cm] (g) at (6.65,3.5) {{$\boldsymbol{\bowtie}$}};
\node[ellipse,draw,minimum width = 0.75cm, minimum height = 0.75cm] (h) at (0.75,4.5) {{$\boldsymbol{\ast}$}};
\node[ellipse,draw,minimum width = 0.75cm, minimum height = 0.75cm] (i) at (2.45,5) {{1}};
\node[ellipse,draw,minimum width = 0.75cm, minimum height = 0.75cm] (j) at (3.35,6.5) {{$\boldsymbol{\ast}$}};
\node[ellipse,draw,minimum width = 0.75cm, minimum height = 0.75cm] (k) at (4.65,6.5) {{1}};
\node[ellipse,draw,minimum width = 0.75cm, minimum height = 0.75cm] (l) at (5.85,5.5) {{1}};
\node[ellipse,draw,minimum width = 0.75cm, minimum height = 0.75cm] (m) at (7.65,5.5) {{0}};


\draw[black, ->] (a) -- (b)node[pos=.35,above, left] {$y$};
\draw[black, ->] (a) -- (c)node[pos=.35,below, right] {$\overline{y}$};
\draw[black, -] (b) -- (d);
\draw[black, -] (b) -- (e);
\draw[black, -] (c) -- (f);
\draw[black, -] (c) -- (g);
\draw[black, ->] (d) -- (h)node[pos=.3,below, left] {$\overline{x}$};
\draw[black, ->] (d) -- (i)node[pos=.3,above, right] {$x$};
\draw[black, ->] (e) -- (j)node[pos=.35,above, left] {$x$};
\draw[black, ->] (e) -- (k)node[pos=.35,above, right] {$\overline{x}$};
\draw[black, ->] (g) -- (l)node[pos=.35,above, left] {$x$};
\draw[black, ->] (g) -- (m)node[pos=.35,above, right] {$\overline{x}$};
\draw[black,dashed] ([xshift=0.75cm,yshift=-0.35cm]b.south) -- ([xshift=-0.3cm,yshift=-0.35cm]b.south)node[pos=.8,above, left] {$T_{7}$};
\draw[black,dashed] ([xshift=-0.75cm,yshift=-0.35cm]c.south) -- ([xshift=0.3cm,yshift=-0.35cm]c.south)node[pos=.8,above, right] {$T_{12}$};
\draw[black,dashed] ([xshift=0.75cm,yshift=-0.25cm]d.south) -- ([xshift=-0.2cm,yshift=-0.25cm]d.south)node[pos=1,above, left] {$T_{6}$};
\draw[black,dashed] ([xshift=0.45cm,yshift=-0.35cm]e.south) -- ([xshift=-0.4cm,yshift=-0.35cm]e.south)node[pos=1,above, left] {$T_{3}$};
\draw[black,dashed] ([xshift=1cm,yshift=-0.25cm]h.south) -- ([xshift=0.2cm,yshift=-0.25cm]h.south)node[pos=1,above, left] {$T_{5}$};
\draw[black,dashed] ([xshift=0.45cm,yshift=-0.35cm]i.south) -- ([xshift=-0.4cm,yshift=-0.35cm]i.south)node[pos=1,above, left] {$T_{4}$};
\draw[black,dashed] ([xshift=1.1cm,yshift=-0.25cm]f.south) -- ([xshift=0.25cm,yshift=-0.25cm]f.south)node[pos=1,above, left] {$T_{11}$};
\draw[black,dashed] ([xshift=-0.5cm,yshift=-0.32cm]g.south) -- ([xshift=0.5cm,yshift=-0.32cm]g.south)node[pos=1,above, right] {$T_{10}$};
\draw[black,dashed] ([xshift=0.4cm,yshift=-0.32cm]j.south) -- ([xshift=-0.3cm,yshift=-0.32cm]j.south)node[pos=.8,above, left] {$T_{2}$};
\draw[black,dashed] ([xshift=-0.35cm,yshift=-0.3cm]k.south) -- ([xshift=0.3cm,yshift=-0.3cm]k.south)node[pos=.8,above, right] {$T_{1}$};
\draw[black,dashed] ([xshift=0.7cm,yshift=-0.3cm]l.south) -- ([xshift=-0.15cm,yshift=-0.3cm]l.south)node[pos=1,above, left] {$T_{8}$};
\draw[black,dashed] ([xshift=-0.6cm,yshift=-0.32cm]m.south) -- ([xshift=0.3cm,yshift=-0.32cm]m.south)node[pos=1,above, right] {$T_{9}$};
\comment{
\draw[black, ->] (m) -- (h);
\draw[black,dashed, ->] (h) -- (g)node[pos=.5,above, left] {$u_1$};

\draw[black, ->] (j.300) edge[bend left]  (e.35);
\draw[black, ->] (k.south) edge[bend left]  (d.east);
\draw[black, ->] (g.south) edge[bend right]  (b.west);}
\end{tikzpicture}

%% file: fig_1.tex
\begin{tikzpicture}[scale=1, transform shape]

\node[ellipse,draw,minimum width = 0.75cm, minimum height = 0.75cm] (a) at (0,0) {{$c_u$}};
\end{tikzpicture}

%% file: fig_2.tex
\begin{tikzpicture}[scale=0.9, transform shape]

\node[diamond,draw,minimum width = 0.75cm, minimum height = 0.75cm] (a) at (0,0) {{$\boldsymbol{\bowtie}$}};
\node[ellipse,draw,minimum width = 0.75cm, minimum height = 0.75cm] (b) at (-1.5,1.5) {{$T^u_j$}};
\node[ellipse,draw,minimum width = 0.75cm, minimum height = 0.75cm] (c) at (1.5,1.5) {{$T^u_k$}};

\draw[black, ->] (a) -- (b)node[pos=.5,below, left] {$\overline{x}$};
\draw[black, ->] (a) -- (c)node[pos=.5,above, right] {$x$};
\end{tikzpicture}

%% file: fig_3.tex
\begin{tikzpicture}[scale=0.9, transform shape]

\node[rectangle,draw,minimum width = 0.75cm, minimum height = 0.75cm] (a) at (0,0) {{$\boldsymbol{\#}$}};
\node[ellipse,draw,minimum width = 0.75cm, minimum height = 0.75cm] (b) at (-1.5,1.5) {{$T^u_j$}};
\node[ellipse,draw,minimum width = 0.75cm, minimum height = 0.75cm] (c) at (1.5,1.5) {{$T^u_k$}};

\draw[black, -] (a) -- (b);
\draw[black, -] (a) -- (c);
\end{tikzpicture}

%% file: lower_bound.tex
\begin{tikzpicture}
                [
point/.style = {circle, draw=black, inner sep=0.07cm,
                ball color=#1,
                node contents={}},
every label/.append style = {font=\footnotesize}
                    ]

\draw[fill=gray!30] (0,0) -- (-2,2) -- (2,2) -- cycle;
\node[] (s) at (0,1.2) {$S'$};
\node[] (a) at (-5,4) {};
\node[] (b) at (5,4) {};
\node[ellipse,draw,minimum width = 0.75cm, minimum height = 0.75cm] (c) at (0,0) {{$\boldsymbol{\bot}$}};
\node[] (d) at (-2,2) {};
\node[] (e) at (2,2) {};

\draw[black, -, shorten <= -2pt, shorten >= -2pt] (a) -- (b)node[pos=1, right] {Axioms};
\draw[black, -, shorten <= -3pt] (a) -- (c);
\draw[black, -, shorten <= -3pt] (b) -- (c);
\draw[black, -, line width=0.5mm, shorten <= -3pt, shorten >= -3pt] (d) -- (e)node[pos=1, below, right] {$S$};
\draw[black, -] (d) -- (c);
\draw[black, -] (e) -- (c);

\node (f) at (0,2) [point=black, label=below:{\scriptsize $L=(C,H^z)$}];
\node (g) at (-1,2.7) [point=black, label=left:{\scriptsize $L_1=(C_1,H^z_1)$}];
\node (h) at (1,2.7) [point=black, label=right:{\scriptsize $L_2=(C_2,H^z_2)$}];

\draw[black, -] (f) -- (g)node[pos=.5, left] {$a_n$};
\draw[black, -] (f) -- (h)node[pos=.5, right] {$\overline{a_n}$};

\draw[black, dashed] (g) -- (-2.5,4) {};
\draw[black, dashed] (g) -- (0.5,4) {};
\node[circle, inner sep=0pt, draw=white, fill=white] (i) at (-1.5,4) {\scriptsize \textbf{$A_{nj}$}};
\node[circle, inner sep=0pt, draw=white, fill=white] (i) at (-0.5,4) {\scriptsize \textbf{$L_A$}};

\node[cross] (j) at (-0.5,4) {};
\end{tikzpicture}